\documentclass[12pt,draftclsnofoot,onecolumn,romanappendices]{IEEEtran}

\usepackage[T1]{fontenc}		
\usepackage[utf8]{inputenc}	

\usepackage{cite}
\usepackage{graphicx}
\usepackage{amsmath,amssymb}
\usepackage{amsfonts}
\usepackage{amsthm}

\usepackage{comment}
\usepackage{epstopdf}
\usepackage{color}

\usepackage{hyperref}
\usepackage{subcaption}
\usepackage{nicefrac}
\usepackage{placeins} 
\usepackage[ruled,vlined,linesnumbered]{algorithm2e}
\usepackage[dvipsnames]{xcolor}

\usepackage{microtype}

\usepackage{makecell}

\newtheorem{remark}{Remark}
\newtheorem{theorem}{Theorem}
\newtheorem{lemma}{Lemma}

\newtheorem{coro}{Corollary}
\newtheorem{definition}{Definition}
\newtheorem{example}{Example}

\usepackage{hyperref} 
\hypersetup{
	colorlinks = true,
	linkcolor  = red,
	citecolor  = green!80!black,
	urlcolor   = black
}

\usepackage{url}

	\newcommand{\Cb}{\mathbb{C}}

	\newcommand{\Gc}{\mathcal{G}}

	\newcommand{\Tc}{\mathcal{T}}

	\DeclareMathOperator*{\argmax}{arg\,max}

\newcommand{\bR}{\Bar{R}}

\definecolor{color1}{rgb}{0.00000,0.44700,0.74100}%
\colorlet{myred}{red!80!black}%
\colorlet{myblue}{color1!80!black}%
\colorlet{green}{green!80!black}%

\newcommand{\cblue}[1]{{\color{blue}{#1}}}

\begin{document}

\title{{\cblue{Benefits of Coded Caching for Multi-antenna Communications with Linear Precoding}}\vspace{-0.1cm}}
\title{Vector Coded Caching Multiplicatively Boosts the Throughput of Realistic Downlink Systems\vspace{-0.1cm}}
\title{{Vector~Coded~Caching~Multiplicatively~Increases} {the Throughput~of~Realistic~Downlink~Systems}\vspace{-0.1cm}}
\author{Hui Zhao, Antonio Bazco-Nogueras, and Petros Elia\vspace{-1.5cm}
	\thanks{Hui Zhao and Petros Elia are with the Communication Systems Department, EURECOM, 06410 Sophia Antipolis, France (email: hui.zhao@eurecom.fr; elia@eurecom.fr). Antonio Bazco-Nogueras is with the IMDEA Networks Institute, 28918 Madrid, Spain (email: antonio.bazco@imdea.org).
	
	This work is supported in part by the Regional Government of Madrid through the grant 2020-T2/TIC-20710 for Talent Attraction, and by the European Research Council under the EU Horizon 2020 research and innovation program/ERC grant agreement no. 725929 (ERC project DUALITY). 
	}
}


\maketitle

	\begin{abstract}\vspace{-0.3cm}
	The recent introduction of vector coded caching has revealed that multi-rank transmissions in the presence of receiver-side cache content can dramatically ameliorate the file-size bottleneck of coded caching and substantially boost performance in error-free wire-like channels. 
	We here employ large-matrix analysis to explore the effect of vector coded caching in realistic wireless multi-antenna downlink systems. Our analysis answers a simple question: Under a fixed set of antenna and SNR resources, and a given downlink MISO system which can already enjoy an optimized exploitation of multiplexing and beamforming gains, what is the multiplicative boost in the throughput when we are now allowed to occasionally add content inside reasonably-sized receiver-side caches?    
	The derived closed-form expressions capture various linear precoders, and a variety of practical considerations such as power dissemination across signals, realistic SNR values, as well as feedback costs. The schemes are very simple (we simply collapse precoding vectors into a single vector), and the recorded gains are notable. For example, for $32$ transmit antennas, a received SNR of $20$ dB, a coherence bandwidth of $300$ kHz, a coherence period of $40$ ms, and under realistic file-size and cache-size constraints, vector coded caching is here shown to offer a multiplicative throughput boost of about $310\%$ with ZF/RZF precoding and a $430\%$ boost in the performance of already optimized MF-based systems. Interestingly, vector coded caching also accelerates channel hardening to the benefit of feedback acquisition, often surpassing 540\% gains over traditional hardening-constrained downlink systems.

   \end{abstract}

	\begin{IEEEkeywords}\vspace{-0.3cm}
	    Coded caching, linear precoding, multi-antenna transmission, random matrix analysis, downlink systems.
	\end{IEEEkeywords}

	\IEEEpeerreviewmaketitle

\section{Introduction}\label{Intro_sec}
\IEEEPARstart{C}{aching} is widely considered to be a valuable resource toward alleviating traffic congestion in various networks~\cite{Paschos,Ciso_forest}. A particularly powerful method for exploiting cache resources can be found in the seminal work of Maddah-Ali and Niesen~\cite{Ali}, who introduced the coded caching framework as a means for exploiting cache-aided side information at the receivers in order to remove interference. This breakthrough was originally presented for the single-stream (single-antenna), error-free, shared-link Broadcast Channel (BC), over which a central server delivers content to $K$ cache-aided users. 
In this context, the server has access to a library of $N$ files, and each user has access to their own dedicated cache of normalized size $\gamma\triangleq\frac{M}{N}\in[0,1]$ corresponding to an individual cache-size equal to the size of $M = \gamma N$ files, and corresponding to a cumulative cache size equal to $K\gamma$ times the size of the library. After a combinatorial content-allocation in each cache during the placement phase, and after each user reveals its demanded file, the delivery phase in~\cite{Ali} employed a novel clique-based scheme that transmitted XORs that could serve $K \gamma+1$ users at a time. 
This astounding multiplicative speed-up factor of $K \gamma+1$ over single-stream cacheless systems was based on the idea that a single XOR carries the desired subfiles of $K \gamma+1$ users, and that these users can utilize their own cached side information to remove undesired subfiles from the XOR in order to recover their own subfile. 
Unfortunately, the clique-based structure of the so-called MN coded caching scheme in~\cite{Ali} requires that the size of each file grows exponentially in $K$ (cf.~\cite{shanmugam2016finite,QifaPDAtit}). This in turn effectively implies --- under realistic file sizes --- a much reduced real speedup factor $\Lambda \gamma+1 \ll K \gamma+1$ for some maximum allowed number of cache-states\footnote{The cache state defines the content stored at the cache of a certain user. Two users sharing the same cache state must store the exact same content in their cache. Having fewer cache-states implies smaller subpacketization and thus smaller required file sizes. A bounded file size forcefully reduces $\Lambda$ as well as the corresponding gain $\Lambda\gamma+1$.} $\Lambda \ll K$. This problem of subpacketization-constrained (or file-size constrained) coded caching is thoroughly documented in a variety of works such as~\cite{QifaPDAtit,HypergraphCodedCachingTit,Krishnan_TIT} as well as~\cite{Ema,Jin,Zhao2021_ISIT,Ibrahim2019}.

At the same time, it also became apparent that for coded caching to develop into an impactful ingredient in wireless systems, it would have to work in conjunction with multi-antenna arrays which are rightfully recognized as the most valuable resource in modern networks. This realization brought to the fore notable research in the area of \emph{multi-antenna coded caching}~\cite{Shariatpanahi_CC,Naderializadeh}, which considers the same model as the aforementioned cache-aided BC, except that now the server (the base-station) is endowed with multiple transmit antennas. In recent years, several related works explored various aspects of the problem, with substantial emphasis on physical-layer considerations. 
One of the first such works can be found in~\cite{Shariatpanahi} which designed physical-layer adaptations of various multi-antenna coded caching schemes. Another interesting approach can be found in~\cite{Bergel2018} which presented a multi-antenna coded-caching scheme for lower SNR regimes when the placement exploits prior information on the users' locations. Furthermore, the work of~\cite{Ngo_TWC} considered the use of transmit antennas for achieving rate scalability in the limit of large $K$, while the work in~\cite{Tolli2020,Meixia_Tao_TWC2019} nicely considered the fusion of multi-antenna multicast beamforming and coded caching toward improved interference management. Interesting work can also be found in~\cite{Bayat,MohammadJavadTWCwithus2021,mohajer2020miso,Karat2019,Salehi2019,Jing,Xu,Vu,Antti9348098,Antti9083779} and in a variety of other publications\nocite{LamEliCachelessTit}. It is the case though that for most of the above schemes, the corresponding degrees-of-freedom (DoF) impact of caching was merely additive to the multiplexing gain (denoted here by $Q$), in the sense that in most of the above scenarios, the DoF performance stagnated at around $Q+\Lambda\gamma$ for very modest values of $\Lambda\gamma$. In essence, \emph{due to the severity of the file-size constraint, the impact of caching was dwarfed by the existing and available multiplexing gains} which have been extensively demonstrated in various field trials\cite{report16layers}.

This imbalance in the impact of caching on multi-antenna systems was reversed with the introduction in~\cite{Lampiris_JSAC} of vector coded caching. This reversal is owed in part to the fact that this new approach could dramatically ameliorate the subpacketization problem previously associated to XOR-based schemes. While previous multi-antenna coded caching techniques essentially focused on using multiple antennas ($L$ transmit antennas) to efficiently deliver the aforementioned sequence of XORs of the original MN scheme, the novel method in~\cite{Lampiris_JSAC} applied a decomposition-based approach that employed a clique structure \emph{on vectors} rather than on scalars. Vector coded caching need not entail the transmission of XORs. Building on the idea of employing $\Lambda$ shared caches ($\Lambda$ cache states) and linear precoding, the algorithm in~\cite{Lampiris_JSAC} was able to offer unprecedented performance as well as a dramatically reduced subpacketization. To be precise, for some $Q\leq L$ representing the aforementioned multiplexing gain of choice, the algorithm in~\cite{Lampiris_JSAC} reduced subpacketization from being exponential in $\Lambda$ to being exponential in $\Lambda/Q$, all while being able to serve up to $Q(1+\Lambda\gamma)$ users at a time. This implied a theoretical multiplicative boost over the DoF of multiplexing-gain systems by a factor of $1+\Lambda \gamma$, with the new DoF of $Q(1+\Lambda\gamma)$ far exceeding the additive impact (see DoF of $Q+\Lambda\gamma)$ of previous XOR-based multi-antenna coded caching approaches. It is the case though that the work in~\cite{Lampiris_JSAC} focused on the error-free, asymptotically high-SNR regime, without considering any practical aspects such as power dissemination across signals, realistic SNR values, the effects of beamforming gain, or the costs of gathering channel state information (CSI). With the exception of some preliminary works like the one in~\cite{Zhao_WSA}, we know very little about the practical performance of vector coded caching in wireless systems. While this new approach was shown to be useful in an information-theoretic (DoF) sense, the real impact that this approach has on optimized downlink systems, has remained an open question.  

 
Any attempt to establish the real impact of vector coded caching must answer a simple question: Under a fixed set of antenna and SNR resources, what is the multiplicative throughput boost obtained from being able to add receiver-side caches to downlink systems that would have otherwise been able to enjoy an optimized exploitation of multiplexing and beamforming gains.    
Indeed, spatial multiplexing and beamforming in multi-antenna downlink systems, and its well-studied application in the large-antenna regime or \emph{massive} multiple-input multiple-output  (MIMO)~\cite{Rusek,LuLu,Emil_unbounded,Rajatheva_whitepaper}, is a key technology in current and future wireless networks that significantly enhances spectral efficiency. Such enhancements have been recently proven in the aforementioned field trials~\cite{report16layers} which demonstrate that a sizeable fraction of the promising theoretic gains brought about by spatial multiplexing approaches, can indeed be attained under practical constraints.

While very considerable research has focused on a variety of advanced precoding schemes, the work-horses of spatial-multiplexing precoding are the optimized versions of linear precoding techniques such as Zero-Forcing (ZF), Regularized ZF (RZF), and Matched Filtering (MF). These techniques maintain low complexity and an ability to provide very high spectral efficiency that often comes close to the optimal performance of the non-linear Dirty-Paper Coding, especially when the number of transmit antennas $L$ is large~\cite{Rusek}. Furthermore, as one would expect, the acquisition of CSI is another ingredient of crucial importance in such systems, even in the presence of Time Division Duplexing (TDD) that partially reduces the CSI overhead as the dimensionality of the problem becomes larger~\cite{Larsson_Mag}. 
This same CSI overhead brings to the fore the issue of channel hardening, which arises as the number of antennas increases, and which partially alleviates the stringent CSI requirements~\cite{Ngo2017hardening}. 

Despite the aforementioned notable research, it is indeed the case that current cellular systems remain under pressure from the increasing user densities and data volumes~\cite{Ciso_forest}. This pressure motivates the search for new resources, and new algorithmic ways to exploit these resources.

\subsubsection*{Structure of Paper and Current Contributions}
The remainder of this paper is organized as follows. We introduce the system model and the considered framework in Section~\ref{sys_sec}. 
Subsequently, in Section~\ref{large_ana_sec}, we first adapt the vector coded caching approach of~\cite{Lampiris_JSAC} to realistic SNR values, while considering three different linear precoding schemes: ZF, RZF and MF. After doing so, we proceed to employ random matrix theory to analyze (in Theorem~\ref{Tight_Bound_MF_coro} for MF, Theorem~\ref{Rate_ZF_lemma} for ZF, and Theorem~\ref{Rate_RZF_lemma} for RZF) the achievable throughput of vector coded caching for the three aforementioned precoders. This analysis --- which naturally incorporates the standard cacheless case corresponding to $\gamma = 0$ --- captures any SNR and any number of users. 

Subsequently, based on the derived asymptotic performance, in Section~\ref{L1_opt_sec} we optimize both the cacheless as well as the cache-aided algorithms by accounting for the CSI acquisition costs, and by optimizing over the total number of simultaneously served streams (users). This optimization, which is performed as a function of SNR, of $L$ and of the CSI acquisition costs, can be found in Theorems~\ref{MF_cost_lemma},~\ref{ZF_cost_lemma}. The same optimization yields systems that are separately calibrated to better balance multiplexing gains with beamforming gains, in the presence or absence of caching. 
In this same section we also derive the ratio between the throughputs of the (independently) optimized cache-aided and cacheless systems. This ratio represents the multiplicative throughput boost offered by caching, over optimized cacheless downlink systems with the same power and antenna resources. 
Subsequently, in Section~\ref{numerical_sec} we numerically verify the accuracy of the derived expressions, showing that they characterize very precisely the actual performance. This evaluation allows us to demonstrate the substantial gains from using caching, highlighting realistic regimes of SNR, $L$, CSI costs, file sizes and cache sizes. In Section~\ref{conclude_sec} we present the main conclusions, while in the appendices we host some of the remaining proofs.

\emph{Notations:} $\mathbb{C}$ stands for the set of complex numbers, ${\bf I}_L \in \mathbb{C}^{L \times L}$ denotes the $L \times L$  identity matrix, and ${\bf 0}_L \in \mathbb{C}^{L \times 1}$ denotes the all-zero vector. We use~$X \sim \mathcal{Y}$ to denote that~$X$ follows the statistical distribution~$\mathcal{Y}$. 
Furthermore, $|\cdot|$  denotes either the cardinality of a set or the magnitude of a complex number, $||\cdot||$ denotes the norm-2 operator for a vector, while we also define $[Z] \triangleq \{1,2,\cdots,Z\}$ for a positive integer $Z$. Additionally, ${\rm Tr}\{\cdot\}$ and $\mathbb{E}\{\cdot\}$ denote the trace and  the expectation operators, respectively, whereas $(\cdot)^T$, $(\cdot)^*$ and $(\cdot)^H$ denote the non-conjugate transpose, conjugate part and conjugate transpose of a matrix, respectively. In asymptotic analysis, $f(x) = o(g(x))$ as $x\to \infty $ denotes that $\lim_{x\to \infty}\frac{f(x)}{g(x)} = 0$. 
Additionally, $\stackrel{a.s.}{\longrightarrow}$ stands for  almost sure convergence. 
If $X \stackrel{a.s.}{\longrightarrow} \mathring{X}$ and $\mathring{X}$ is deterministic, we call $\mathring{X}$  the asymptotic deterministic equivalent of $X$. Moreover, in the limit of $x\to \infty$, our use of $A(x) \doteq B(x)$ will mean that $A(x) = B(x) + o(1)$.

\section{System Model and Problem Description}\label{sys_sec}

\subsection{System Model}\label{JSAC_Intro}
We consider a downlink MISO scenario where an $L$-antenna base station (BS) 
serves $K$ single-antenna cache-aided users. The BS has access to a library 
of $N$ equally-sized files, and each user is endowed with a local memory (or cache) of size equal to the size of $M$ library files ($M < N$), such that each user can store a fraction $\gamma = \frac{M}{N} \in [0,1)$ of the library content. We denote the library content by $\mathcal{F}$ and the $n$-th file by $W_n$, such that $\mathcal{F} \triangleq \{W_n\}_{n=1}^N$.

We consider the wireless channel to be modeled as a symmetric Rayleigh fading channel, where all channel coefficients are assumed to be independent and identically distributed (i.i.d.). 
When describing a general transmission, our notation will often incorporate the subset $\mathcal{K} \subseteq [K]$ of users that are simultaneously served during that transmission. Consequently, in our communication model, the received signal at the $k$-th user in $\mathcal{K}$ is given by
\vspace{-0.2cm}
\begin{align}
    y_{\mathcal{K}(k)} = {\bf h}_{\mathcal{K}(k)}^T {\bf x}_{\mathcal{K}} + z_{\mathcal{K}(k)},\vspace{-0.2cm}
\end{align}
where  $k \in [|\mathcal{K}|]$, where $z_{\mathcal{K}(k)} \in \mathbb{C}$ represents the corresponding Additive White Gaussian Noise (AWGN) with zero-mean and unit-variance, where ${\bf x}_{\mathcal{K}}\in \mathbb{C}^{L\times 1}$ denotes the transmitted signal vector that simultaneously serves the users in $\mathcal{K}$, and where ${\bf h}_{\mathcal{K}(k)} \in \mathbb{C}^{L\times 1}$ represents the channel vector for the channel from the BS to the $k$-th user in $\mathcal{K}$. As mentioned, ${\bf h}_{\mathcal{K}(k)}$ is assumed to be an i.i.d. Gaussian random vector with mean ${\bf 0}_{L}$ and covariance matrix ${\bf I}_{L}$. Finally, ${\bf x}_{\mathcal{K}}$ is obtained by applying a specific precoding scheme (which we will detail later on) to the information vector ${\bf s}_{\mathcal{K}} \in \mathbb{C}^{|\mathcal{K}| \times 1}$ intended for the users in $\mathcal{K}$, where ${\bf s}_{\mathcal{K}}$ has mean ${\bf 0}_{|\mathcal{K}|}$ and covariance matrix ${\bf I}_{|\mathcal{K}|}$.

We consider an average power normalization, where the power is averaged over both transmit symbols and channel realizations. As is common in practical downlink settings, we assume TDD uplink-downlink transmissions, such that the BS estimates the downlink channels through uplink pilot transmissions by applying channel reciprocity. 
 
We proceed to describe the main structure of the scheme, first doing so without specifying the linear precoding class that is used. We will also formally define the main performance metrics investigated in this paper.

\subsection{Signal-Level Vector Coded Caching for Finite SNR}\label{JSAC_Intro}
		Building on the general vector-clique structure in~\cite{Lampiris_JSAC}, we are here free to choose the precoding schemes, as well as calibrate at will the dimensionality of each vector clique. This freedom is essential in controlling the impact of CSI costs and of power-splitting across users, both of which directly affect the performance in practical SNR regimes. \nocite{AdeelTCOM21}
			
		 We proceed to describe the cache placement phase and the subsequent delivery phase.
		 

\subsubsection{Placement Phase}

The first step involves the partition of each library file $W_n$ into ${\Lambda \choose \Lambda \gamma}$ non-overlapping equally-sized subfiles
$\big\{ W_n^\mathcal{T}:  \mathcal{T} \subseteq [\Lambda], |\mathcal{T}|=\Lambda \gamma  \big\}$, each labeled by some $\Lambda\gamma$-tuple $\mathcal{T}\subseteq[\Lambda]$. As discussed in Section~\ref{Intro_sec}, the number of cache states $\Lambda$ is chosen to satisfy the file-size constraint\footnote{In our case, the subpacketization is ${\Lambda \choose \Lambda \gamma}$, which naturally serves as a lower bound on the file sizes.}.   
Subsequently the $K$ users are \emph{arbitrarily} separated into $\Lambda$ disjoint groups $\mathcal{G}_1,\mathcal{G}_2,\dots,\mathcal{G}_\Lambda$, where the $g$-th group, which consists of $B = \frac{K}{\Lambda}$ users, is given by $ \mathcal{G}_g \triangleq \big\{b \Lambda +g\big\}_{b=0}^{B-1} \subseteq [K]$. The $\vartheta$-th user of this $g$-th group is denoted\footnote{We will henceforth consider $K$ to be a multiple of $\Lambda$. This assumption is adopted for the sake of clarity of exposition, and it does not limit the scope of the results in any way. The general case can be readily handled (cf.~\cite{Lampiris_JSAC}). Moreover, this grouping as well as the entire placement phase, are naturally done before the users' requests take place, and of course well before the channel states are known to the BS.} by ${\rm U}_{g,\vartheta}$.

At this point, all the users belonging to the same group are assigned the same cache state and thus proceed to cache \emph{identical} content. In particular, for those in the $g$-th group, this content takes the form
$\mathcal{Z}_{\Gc_g} = \big\{ W_n^{\mathcal{T}}:  \Tc \ni g,\, \forall n\in[N]  \big\}$.

\subsubsection{Delivery Phase}

	This phase starts when each user $\kappa \in [K]$ simultaneously asks for its intended file, denoted here by $W_{d_\kappa}$, $d_\kappa \in [N]$. 
	The BS selects $Q$ users from each group, where $Q\leq B$ is a variable that will be optimized afterwards and which is the equivalent of the multiplexing gain. By doing so, the BS decides to first `encode' over the first $\Lambda Q$ users, and to repeat the encoding process $B/Q$ times\footnote{To clarify, what the above says is the following. If there are, e.g., $B=2Q$ users per group and thus $K = 2 \Lambda Q$ users in total, then the algorithm that we describe here will be first applied to the first $\Lambda Q$ users, and then, after this delivery is done, the same algorithm will apply to the remaining $\Lambda Q$ users, thus eventually satisfying all $K$ users. Also note that a small amount of additional subpacketization can easily resolve the case where $B/Q$ may not be an integer.}. 
	To deliver to the $\Lambda Q$ users, the transmitter employs ${\Lambda \choose \Lambda \gamma+1}$ sequential transmission stages. During each such stage, the BS simultaneously serves a unique set $\Psi$ of $|\Psi| = \Lambda \gamma+1$ groups, corresponding to a total of 
	$Q (\Lambda \gamma+1)$ users served at a time (i.e., per stage).
    At the end of the ${\Lambda \choose \Lambda \gamma+1}$ transmission stages, all the $\Lambda Q$ users obtain their intended files. By repeating this process $\big\lceil\frac{B}{Q}\big\rceil$ times, all the $K$ users obtain their intended files.
    As suggested above, the factor $G \triangleq \Lambda\gamma+1$ describes the number of user-groups that are served simultaneously. Another crucial parameter includes the multiplexing gain $Q$ which, unlike in~\cite{Lampiris_JSAC}, will be here subject to optimization.

	Let us now focus on a single transmission stage. As mentioned above, at each such stage, we pick a set $\Psi\subseteq \Lambda$ of $G =\Lambda \gamma+1$ groups that will be served simultaneously. From within these chosen groups, we will serve $Q \leq B$ users per group.
	In particular, for each user ${\rm U}_{\psi,\vartheta}$ of some group $\psi\in \Psi$, this stage will deliver all subfiles\footnote{In a slight abuse of notation, we use the term ``subfile" to refer both to the actual  subfile generated after file-splitting, as well as to the corresponding complex-valued information symbol $s_{\psi,\vartheta}$.} $s_{\psi,\vartheta}$ by transmitting
	\begin{align}\label{eq:transmitSignal1}
			{\bf x}_\Psi &= \frac{1}{\sqrt{G}} \sum\nolimits_{\psi \in \Psi} \rho_\psi  \sum\nolimits_{\vartheta=1}^{Q}{\bf v}_{\psi,\vartheta} s_{\psi,\vartheta},
	\end{align}
	where ${\bf v}_{\psi,\vartheta} \in \mathbb{C}^{L \times 1}$  denotes the precoder applied to the subfile intended by user ${\rm U}_{\psi,\vartheta}$, and where $\rho_{\psi}$ denotes the power normalization factor for group $\psi\in\Psi$, applied under a total power constraint $P_t$. 
	Upon defining $ {\bf V}_{\psi}\in\Cb^{L\times Q}$ as $ {\bf V}_{\psi} \triangleq [{\bf v}_{\psi,1} \big| \dotsc \big| {\bf v}_{\psi,Q}]$ and ${\bf s}_\psi \in \Cb^{Q}$ as ${\bf s}_\psi \triangleq [s_{\psi,1}, \dotsc, s_{\psi,Q}]^T$, the above takes the simple form	
	\begin{align}\label{eq:transmitSignal2}
			{\bf x}_\Psi 
			&= \frac{1}{\sqrt{G}} \sum\nolimits_{\psi \in \Psi} \rho_\psi  {\bf V}_{\psi} {\bf s}_\psi.
	\end{align}
	\begin{remark}
	    {{It is easy to see that the described scheme simply involves a carefully selected linear combination of $G$ linear-precoding vectors that are now to be sent simultaneously.  It is also easy to see that the above scheme also incorporates the traditional cacheless downlink scenario corresponding to $\gamma = 0$ which itself corresponds to $G = |\Psi|=1$. In such case, the transmit signal expression reverts to the simpler common expression 
	    ${\bf x} 
			= \rho {\bf V}{\bf s}$. 
	    }}
	\end{remark}
	For decoding to work, the subfiles must be chosen carefully. This choice follows the principles of coded caching, and in particular of vector coded caching. Thus, when considering the transmission stage which serves the $G = \Lambda\gamma+1$ groups in $\Psi$, the subfile transmitted to user  ${\rm U}_{\psi,\vartheta}$ is here selected to be {$ W_{d_{\psi,\vartheta}}^{\Psi \setminus \{\psi\}}$, simply because this subfile is stored in the cache of each user of every other group in $\Psi$ except $\psi$.} 
	Because of this structure, the users of a particular group can remove the inter-group interference from the other $\Lambda \gamma$ groups by using their cached content. 
	On the other hand, following the principles of vector coded caching, the intra-group interference is handled with linear precoding that `separates' the signals of the users from the same group.  
	Naturally one can imagine that cache-aided removal of interference as well as `nulling out' of interference, both require 
	knowledge of the composite precoder-channel coefficients (cf. \eqref{y_psi_k_eq_initial}). These so-called \emph{composite CSI} costs will be explicitly accounted for in our analysis.
	We proceed to elaborate on the precoders and the transmissions. 

\vspace{-0.25cm}
\subsection{Vector Coded Caching for the Physical Layer}
We now emphasize on the physical layer details of the communication scheme. Our description will focus on the transmission that serves a specific set~$\Psi$ of user-groups. First let us recall that ${\bf V}_\psi \in \mathbb{C}^{L \times Q}$ denotes the precoding matrix for the symbols\footnote{We note that as is common, our analysis will assume Gaussian signaling.} of users in group~$\psi \in \Psi$.  Then let us note that for an average power constraint~$P_t$, the power normalization factor~$\rho_\psi$ from~\eqref{eq:transmitSignal2}, takes the form 
	$\rho_\psi = \sqrt{\frac{P_t}{\mathbb{E}\{ {\bf s}_\psi^H {\bf V}_\psi^H {\bf V}_\psi {\bf s}_\psi \}}} = 
	 \sqrt{\frac{P_t}{ \mathbb{E} \{{\rm Tr} \{ {\bf V}_\psi^H {\bf V}_\psi \} \}}}$. 
Then the subsequent corresponding received signal at user ${\rm U}_{\psi,k}$ (i.e., at the $k$-th user of group $\psi \in \Psi$), will take the form
    \begin{align}\label{eq:def_intergroup_int}
        y_{\psi,k} & ={\bf h}_{\psi,k}^T {\bf x}_{\Psi} + z_{\psi,k} 
            = \frac{{\bf h}_{\psi,k}^T }{\sqrt{G}} \rho_\psi {\bf V}_{\psi} {\bf s}_{\psi}
            +\underbrace{ \frac{{\bf h}_{\psi,k}^T}{\sqrt{G}}  \sum\nolimits_{\phi \in \Psi,\phi \neq \psi}\rho_\phi {\bf V}_{\phi} {\bf s}_{\phi} }_{\text{inter-group interference}} + z_{\psi,k}.
    \end{align}
As previously mentioned, the inter-group interference\footnote{As a reminder, the term inter-group interference refers to the received signal component whose power is due to the information meant for users originating from other groups.} experienced by user ${\rm U}_{\psi,k}$, can be removed from~$y_{\psi,k}$ by exploiting that same user's cached content and that user's composite CSI $\{{\bf h}_{\psi,k}^T{\bf v}_{\phi,k'}\rho_\phi\}_{\phi\in\{\Psi\setminus\psi\},\, k'\in[Q]}$.
Then, after the cache-aided removal of this inter-group interference, the equivalent received signal at ${\rm U}_{\psi,k}$ is given by
\begin{align}\label{y_psi_k_eq_initial}
		y_{\psi,k}^\prime 
		= \frac{\rho_\psi}{\sqrt{G}} {\bf h}_{\psi,k}^T {\bf v}_{\psi,k} {s}_{\psi,k} + \underbrace{\frac{\rho_\psi}{\sqrt{G}} \sum\nolimits_{\vartheta=1,\vartheta \neq k}^{Q} {\bf h}_{\psi,k}^T {\bf v}_{\psi,\vartheta} {s}_{\psi,\vartheta}}_{\text{intra-group interference}} {} + z_{\psi,k}.
\end{align}
Consequently, the corresponding SINR for information decoding at ${\rm U}_{\psi,k}$, is given by
\begin{align}\label{SINR_actual_def}
    {\rm SINR}_{\psi,k} = \frac{\frac{\rho_\psi^2}{G} |{\bf h}_{\psi,k}^T {\bf v}_{\psi,k}|^2}{1+\frac{\rho_\psi^2}{G} \sum\nolimits_{\vartheta=1,\vartheta \neq k}^{Q} |{\bf h}_{\psi,k}^T {\bf v}_{\psi,\vartheta}|^2}.
\end{align}

On the other hand, in the cacheless case of $\gamma = 0$, the received signal $y_{k}	= \rho{\bf h}_{k}^T {\bf v}_{k} {s}_{k} + \rho \sum\nolimits_{\vartheta=1,\vartheta \neq k}^{Q} {\bf h}_{k}^T {\bf v}_{\vartheta} {s}_{\vartheta} + z_{k}$ at some user $k$ naturally carries no inter-group interference (as there are no other groups to simultaneously serve), and the SINR takes the standard form ${\rm SINR}_{k} = \frac{\rho^2|{\bf h}_{k}^T {\bf v}_{k}|^2}{1+\rho^2 \sum\nolimits_{\vartheta=1,\vartheta \neq k}^{Q} |{\bf h}_{k}^T {\bf v}_{\vartheta}|^2}.$

We will consider the MF, ZF and RZF linear precoding schemes, selected here for being very common, simple, as well as competitive in terms of rate performance~\cite{Peel,Paulraj}. 
  As is known, the corresponding precoding matrices ${\bf V}_\psi$ take the form:
\begin{align}\label{eq:precoders}
	{\bf V}_\psi = 
		\begin{cases}    
			 {\bf H}_\psi^H, & \text{ MF Precoder}\\
			{\bf H}_\psi^H \left( {\bf H}_\psi {\bf H}_\psi^H \right)^{-1}, & \text{ ZF Precoder} \\
			{\bf H}_\psi^H \left( {\bf H}_\psi {\bf H}_\psi^H + \alpha {\bf I}_{Q}\right)^{-1}, & \text{ RZF Precoder},
			\end{cases}
\end{align}
where ${\bf H}_\psi \triangleq \big[ {\bf h}_{\psi,1} | {\bf h}_{\psi,2} | \cdots | {\bf h}_{\psi,Q} \big]^T \in \mathbb{C}^{Q \times L}$ denotes the channel matrix for the channel between the BS to the $Q$ chosen users\footnote{As expected, $Q$ must always be no larger than $B$, while in the case of the ZF/RZF precoder this value must also be bounded as $Q\leq L$.} belonging to group $\psi \in \Psi$, and where $\alpha$ is the regularization factor of the RZF precoder\footnote{It is worth recalling that the RZF precoder reverts to the ZF precoder when $\alpha=0$, and to the MF precoder when $\alpha \to \infty$.}~\cite{Peel}. 
For simplicity we assume that $\alpha = L/P_t$, which is a commonly used assumption throughout the literature~\cite{Peel,Wagner,Hoydis}.

We will henceforth use the term $(G,Q)$-vector coded caching, to refer to the vector coded caching scheme when it serves $G$ groups with $Q$ users per group. We will also use the term \emph{MF-based $(G,Q)$-vector coded caching} to refer to the same scheme when the underlying precoder is MF, and similarly we will use ZF-based or RZF-based $(G,Q)$-vector coded caching, for the other two precoders.
Let us now formally define some important metrics of interest.

\begin{definition}\label{sum_rate_def}\vspace{-0.3cm}
\emph{(Average sum-rate and effective sum-rate).} For a $(G,Q)$-vector coded caching scheme, its average sum-rate is denoted by $\bar R(G,Q)$ and is defined as the total data-transmission rate (before accounting for CSI costs) summed over the $GQ$ simultaneously served users, and averaged over the fading.  Similarly, the effective average sum-rate $\mathcal{\bar R}(G,Q)$ will represent the corresponding average rate after through all CSI costs are duly accounted for. 
\end{definition}
\begin{definition}\label{effective_def}\vspace{-0.3cm}
\emph{(Effective gain over MISO).} For a given set of $L$ and SNR resources, and a fixed underlying precoder class, the effective gain, after accounting for CSI costs, of the $(G,Q)$-vector coded caching over the cacheless scenario (corresponding to $G = 1$, and an operating multiplexing gain $Q'$), will be denoted as
$\Gc(G, Q; 1,Q') \triangleq \frac{\mathcal{\bar R}(G,Q)}{\mathcal{\bar R}(1,Q')}$ in the form of the ratios of the~\emph{effective  rates}.
\end{definition}\vspace{-0.3cm}

  
\vspace{-0.2cm}\section{Analysis of the Average Rate and of the Effective Gain over MISO}\label{large_ana_sec}

In this section, we analyze the average sum-rates and the corresponding effective rates achieved by the cache-aided downlink schemes of Section~\ref{JSAC_Intro} for the MF, ZF and RZF linear precoders of interest. After doing so, we also report the effective gains offered by these $(G,Q)$-vector coded caching schemes, over the $(G=1,Q')$ cacheless equivalents.  

We will henceforth consider the ratio $c \triangleq Q /L$, while we will often use the notation $c' \triangleq Q' /L$ when referring explicitly to the cacheless equivalent. The two ratios can be chosen independently. When applying large matrix analysis, we will be assuming a fixed $c>0$ and a fixed $c'>0$. 

\vspace{-0.3cm}
\subsection{MF Precoding}
To derive the average sum-rate of vector coded caching with MF precoding, we first recall that the elements of ${\bf H}_\psi$ are i.i.d. Gaussian random variables with zero mean and unit variance, which implies that $\mathbb{E}\left\{ {\rm Tr} \left\{ {\bf H}_\psi  {\bf H}_\psi^H \right\} \right\}=L Q$ (cf.~\cite{Matthaiou}), which then implies that the power normalization factor $\rho_\psi$ takes the form
		$\rho_\psi = \sqrt{\frac{P_t}{ \mathbb{E} \left\{ {\rm Tr}\left\{ {\bf H}_\psi {\bf H}_\psi^H \right\} \right\}}} = \sqrt{\frac{P_t}{Q L}}$ (cf.~\cite{Feng}).
	This in turn yields (cf.~\eqref{eq:precoders},~\eqref{eq:transmitSignal2}) a transmitted signal of the form
	\begin{align}\label{eq:transmitSignal_mf}
			{\bf x}_\Psi =
				\sqrt{\frac{P_t}{G Q L} } \sum\limits_{\psi \in \Psi}  {\bf H}_\psi^H {\bf s}_{\psi} = \sqrt{\frac{P_t}{G Q L} } \sum\limits_{\psi \in \Psi} \sum\limits_{\vartheta=1}^Q {\bf h}_{\psi,\vartheta}^* s_{\psi,\vartheta}.
		\end{align}
The corresponding average sum-rate is presented below. 


\begin{theorem}\label{Tight_Bound_MF_coro}
     For any given $P_t$ and $c=Q/L$, the average sum-rate~$\bar R^{\rm MF}$ of the MF-based $(G,cL)$-vector coded caching scheme in the large $L$ regime, takes the form
    \begin{align}\label{R_MF_sum_caseii}
        \bar{R}^{\rm MF}(G,cL) \doteq c\  G L   \ln \left( 1+ \frac{1}{c} \frac{P_t }{P_t  + G } \right).
    \end{align}
\end{theorem}

\begin{proof}
The proof can be found in Appendix~\ref{Proof_Thm1_Zhang}.
\end{proof}

The following directly distills the above result to the cacheless case\footnote{It is worth noting that while there have been various works (cf.~\cite{Hien,Hoydis,Yeon_Geun_Lim}) analyzing the MF sum-rate in traditional massive MIMO systems, the result derived in this work here entails less assumptions. For example, focusing on the large-$L$ regime, the result in~\cite{Hien} directly assumes a tight Jensen's bound, while the result in~\cite{Yeon_Geun_Lim} is under a so-called ``near deterministic'' assumption in low/high SNRs. On the other hand, our method here draws from the uplink analysis in~\cite{Zhang}, and only employs a large-$L$ assumption to derive the exact asymptotic optimality for any value of SNR.}. 
\begin{coro}\label{cor:EffectiveRatesCacheless}
     In the limit of large $L$, and for any fixed $P_t$ and $c'$, the average sum-rate of the (traditional, cacheless) MF-based MISO BC with $c' L$ streams takes the form
    \begin{align}\label{R_MF_sum_caseiiCacheLess}
        \bar{R}^{\rm MF}(1,c' L) \doteq c' \  L   \ln \left( 1+ \frac{1}{c'} \frac{P_t }{P_t  + 1 } \right).
    \end{align}
\end{coro}

\vspace{-0.5cm}
\subsection{ZF Precoding}
		Moving now to the case of ZF-based vector coded caching, and focusing again on a set of groups $\Psi$ and on the transmission stage corresponding to some group $\psi\in \Psi$, we have that the power control factor takes the form
				$\rho_\psi^2  = 
				{\frac{P_t}{\mathbb{E} \{  {\rm Tr} \{ ( {\bf H}_\psi {\bf H}_\psi^H )^{-1}  \} \}}}$,
		while the transmitted signal from~\eqref{eq:transmitSignal2} becomes
			\begin{align}
			{\bf x}_\Psi 
			= \frac{1}{\sqrt{G}} \sum\limits_{\psi\in\Psi} \rho_\psi {\bf H}_\psi^H \left( {\bf H}_\psi {\bf H}_\psi^H \right)^{-1} {\bf s}_{\psi}.
			\end{align}
		This in turn yields a received signal at user ${\rm U}_{\psi,k}$ which --- after the cache-aided removal of the inter-group interference (cf.~\eqref{eq:def_intergroup_int}) --- takes the form
		\begin{align}
		y_{\psi,k}^\prime &=  \frac{1}{\sqrt{G}}  \rho_\psi {\bf h}_{\psi,k}^T {\bf H}_\psi^H \left( {\bf H}_\psi {\bf H}_\psi^H \right)^{-1} {\bf s}_{\psi} + z_{\psi,k} \notag\\
		&= \frac{1}{\sqrt{G}}  \rho_\psi \left( {\bf 1}_{k}^T {\bf H}_\psi \right) {\bf H}_\psi^H \left( {\bf H}_\psi {\bf H}_\psi^H \right)^{-1} {\bf s}_{\psi} + z_{\psi,k} 
		= \frac{1}{\sqrt{G}} \rho_\psi s_{\psi,k} + z_{\psi,k},
		\end{align}
		where   ${\bf 1}_k \in \mathbb{C}^{Q \times 1}$ denotes the vector whose components are all zero except for the $k$-th element, which equals~$1$.  
		After then considering that all intra-group interference is canceled by means of ZF precoding, we can write the SINR at user ${\rm U}_{\psi,k}$ as 
			\begin{align}\label{SNR_ZF_def}
				{\rm SINR}_{\psi,k}^{\rm ZF} = 
				\frac{P_t}{G\, \mathbb{E} \big\{ {\rm Tr}\big\{\left( {\bf H}_\psi {\bf H}_\psi^H \right)^{-1} \big\} \big\} }.
			\end{align}
		With this in place, we proceed with the following theorem. 
		\vspace{-0.3cm}
			\begin{theorem}\label{Rate_ZF_lemma}
        {{For $c = \frac{Q}{L} \in(0,1)$}}, the average sum-rate~$\bar R_{\rm sum}^{\rm ZF}$ of the ZF-based $(G,Q)$-vector coded caching scheme, takes the form
				    \begin{align}\label{ZF_sumrate}
					     \bar{R}^{\rm ZF}(G,Q) = QG \ln \left( 1+ \frac{P_t}{G} \left( \frac{1}{c} -1 \right) \right).    
					\end{align}
			\end{theorem}

		\begin{proof}
			Directly from~\cite{Tulino_matrix}, and from the fact that ${\bf H}_\psi {\bf H}_\psi^H$ is a Wishart matrix with $L$ degrees of freedom, we know that $\mathbb{E}\big\{ {\rm Tr}\big\{\big( {\bf H}_\psi {\bf H}_\psi^H \big)^{-1} \big\} \big\} = \frac{Q}{L-Q}$ for $L>Q$. 
			Naturally, ${\rm SINR}_{\psi,k}^{\rm ZF}$ is deterministic and constant across all simultaneously served users. {{By summing the average rate of each of the  $G Q$ served users, we obtain~\eqref{ZF_sumrate}.}}
		\end{proof}

\vspace{-0.5cm}
\subsection{RZF Precoding}
We finally consider our third precoder, and do so in the asymptotic regime of large $L$ and fixed $c$. 
	We first note that the received signal at  ${\rm U}_{\psi,k}$ --- after cache-aided removal of the inter-group interference --- takes the form
    \begin{align}\label{y_RZF_Exp}
		y'_{\psi,k} &= 
		 \frac{\rho_\psi}{\sqrt{G}} \sum_{\vartheta=1}^{Q} {\bf h}_{\psi,k}^T \left(\alpha {\bf I}_L + {\bf H}^H_\psi {\bf H}_\psi \right)^{-1} {\bf h}_{\psi,\vartheta}^* { s}_{\psi,\vartheta} + z_{\psi,k}.
    \end{align}
    For ${\bf H}_{\psi,-k}$ denoting the matrix resulting from  ${\bf H}_{\psi}$ after removing its $k$-th row, we proceed to define 
		\begin{align}
            &A_{\psi,k} \triangleq {\bf h}_{\psi,k}^T \Big( \alpha {\bf I}_L + {\bf H}_{\psi,-k}^H {\bf H}_{\psi,-k}\Big)^{-1} {\bf h}_{\psi,k}^* \label{def_a}, \\
			&B_{\psi,k} \triangleq {\bf h}_{\psi,k}^T \left(\alpha {\bf I}_L + {\bf H}_{\psi,-k}^H {\bf H}_{\psi,-k} \right)^{-1}{\bf H}_{\psi,-k}^H {\bf H}_{\psi,-k} \big(\alpha {\bf I}_L + {\bf H}_{\psi,-k}^H {\bf H}_{\psi,-k}\big)^{-1} {\bf h}_{\psi,k}^*.\label{def_b}
		\end{align}
	With these in place, we can now derive the SINR at user ${\rm U}_{\psi,k}$ to be	
		\begin{align}\label{SINR_L1_def}
			{\rm SINR}_{\psi,k}^{\rm RZF} 
			=\frac{{A}_{\psi,k}^2\frac{{\rho}_\psi^2}{G} }{\big(1+{A}_{\psi,k} \big)^2+\frac{{\rho}_\psi^2}{G}{B}_{\psi,k}},
		\end{align}
	where the proof of~\eqref{SINR_L1_def} is relegated to Appendix~\ref{app_sinr}. 
	
	We can now present the asymptotic deterministic equivalent of the sum-rate of our proposed scheme when RZF is applied.
	We recall that in the limit of large $L$, the deterministic value $\mathring{X}$ represents the \emph{asymptotic deterministic equivalent} of $X$ if $X \stackrel{a.s.}{\longrightarrow} \mathring{X}$.

	\begin{theorem}\label{Rate_RZF_lemma}
		In the large-$L$ regime with fixed $c = Q/L$, the average sum-rate $\bR^{\rm RZF}$ of RZF-based $(G,Q)$-vector coded caching, takes the form 
			\begin{align}\label{Rate_RZF_eq}
			\bR^{\rm RZF}(G,Q)\doteq \mathring{R}^{\rm RZF}(G,cL)
			\triangleq c\, G L \ln \left(1+ \frac{a_{\psi,k}^2 p_\psi^2 / G}{\big(1+a_{\psi,k} \big)^2+P_t/G}\right),     
			\end{align}
		where $\mathring{R}^{\rm RZF}$ is the deterministic equivalent of $\bR^{\rm RZF}$,\footnote{This entails a small abuse of terminology, as it is $\mathring{R}^{\rm RZF}/L$ that is the deterministic equivalent of $\bR^{\rm RZF}/L$.} and where  
				\begin{align}
					a_{\psi,k} & \triangleq \ \  \frac{1}{2}  \left[  \sqrt{ (1-c)^2 P_t^2 + 2 (1+c) P_t +1} +  (1-c)P_t -1  \right],\hspace{7ex} \label{A_psi_K_deter} \\
					 p_\psi^2 &\triangleq  \ \ \frac{P_t}{a_{\psi,k} - \frac{P_t}{2}\Big(\frac{P_t(c-1)^2+c+1}{\sqrt{P_t^2(c-1)^2 + 2(c+1)P_t+1}} +  {1-c}\Big)}.\label{rho_psi_K_deter}
				\end{align}			
	\end{theorem}

	\begin{proof}
	The proof is based on the derivation of the asymptotic deterministic equivalent of the SINR, and it is presented in Appendix~\ref{proof_theo_asymp}.
	\end{proof}

\subsection{Accounting for the CSI Costs}
	
To account for the cost of CSI acquisition under TDD, we consider a basic CSI-acquisition effort where at the beginning of each transmission stage, the $G Q$ served users send uplink orthogonal pilot symbols, from which the BS can estimate the downlink channel matrix, under the assumption of channel reciprocity. Then the CSI-acquisition process engages downlink training, of similar complexity, in order to communicate the composite CSI that here allows our receivers to perform cache-aided cancellation of the inter-group interference (cf.~\eqref{eq:def_intergroup_int}) from their signal\footnote{This acquisition process for gathering composite CSI, with the same aforementioned complexity per served user, is standard in a variety of traditional communications techniques such as SIC-based approaches. For additional details, please see~\cite{lampiris2021_Bottleneck}.}. To account for this CSI-acquisition overhead, we directly extend the commonly-used approach in~\cite{Kobayashi2012, Sadeghi2018_multicast}, that easily allows us to calculate the effective average sum-rate (cf. Definition~\ref{sum_rate_def}) for each precoder $i\in\{$MF, ZF, RZF$\}$, to be
    \begin{align}\label{Rate_CSI_Def}
     {\mathcal{\bar R}}^i  = \left( 1 - \frac{\beta_{\rm tot} G Q}{T_c W_c} \right)  \bar R^i
     = \left( 1 - c \zeta_{G,Q} \right) \bar R^i,
    \end{align}
    where $\beta_{\rm tot}$ is the number of resources per user and per block used for pilot transmission,  $\bar R^i$ is the previously calculated average sum-rate before accounting for CSI costs, where $T_c$ and $W_c$ are the coherence time and coherence bandwidth, respectively, and where $\zeta_{G,Q} \triangleq \frac{\beta_{\rm tot} G L}{T_c W_c}$. For completeness we report the effective rates in the following corollary. The proof is direct as it merely involves applying \eqref{Rate_CSI_Def} in the expressions from Theorems~\ref{Tight_Bound_MF_coro}-\ref{Rate_RZF_lemma}. We recall that $a_{\psi,k}$ and $p_\psi$ are defined in Theorem~\ref{Rate_RZF_lemma}.
    
\begin{coro}\label{cor:EffectiveRates3}
    The effective rates of the proposed vector coded caching schemes under MF, ZF and RZF precoding, respectively take the form
    \begin{align}\label{effect_gain_MF_eq}
        {\mathcal{\bar R}}^{\text{MF}}(G,Q)  & \doteq \left( 1 - c \zeta_{G,Q} \right)  c\  G L   \ln \left( 1+ \frac{1}{c} \frac{P_t }{P_t  + G }\right), \\
        {\mathcal{\bar R}}^{\text{ZF}}(G,Q)  & = \left( 1 - c \zeta_{G,Q} \right)  QG \ln \left( 1+ \frac{P_t}{G} \left( \frac{1}{c} -1 \right) \right), \\
        {\mathcal{\bar R}}^{\text{RZF}}(G,Q)  & \doteq  \left( 1 - c \zeta_{G,Q} \right)  c\, G L \ln \left(1+ \frac{a_{\psi,k}^2 p_\psi^2 / G}{\big(1+a_{\psi,k} \big)^2+P_t/G}\right).
    \end{align} 
\end{coro}


\subsection{Effective Gains over Cacheless MISO Systems}
At this point, with Theorems~\ref{Tight_Bound_MF_coro},~\ref{Rate_ZF_lemma},~\ref{Rate_RZF_lemma} in place, and in conjunction with Corollary~\ref{cor:EffectiveRates3}, we can directly report the effective gains over cacheless MISO. For each of the three precoder classes, MF, ZF, and RZF, and for a fixed set of antenna and SNR resources, we will be reporting the effective gain 
$\Gc(G, Q; 1,Q') = \frac{\mathcal{\bar R}(G,Q)}{\mathcal{\bar R}(1,Q')}$ (cf.~Definition~\ref{effective_def}) of the $(G,Q)$-vector coded caching schemes, over the cacheless scenario ($G = 1$) with some chosen number of streams $Q'$.
These effective gains are collected together in the following corollary. 

\begin{coro}\label{cor:gains}
    The effective gains of the proposed vector coded caching schemes under MF, ZF and RZF precoding, respectively take the form
    \begin{align}\label{effect_gain_MF_eq}
        &{\Gc}_{\rm MF} \left(G,Q;1,Q'\right) \triangleq \frac{{\mathcal{\bar R}}^{\text{MF}}(G,Q)}{{\mathcal{\bar R}}^{\text{MF}}(1,Q')} 
        \doteq \xi \frac{G Q}{ Q'} \frac{\ln \left( 1+\frac{L}{Q} \frac{P_t}{P_t+G} \right)}{\ln \left( 1+\frac{L}{Q'} \frac{P_t}{P_t+1} \right) },\\
    	&{\Gc}_{\rm ZF} \left(G,Q;1,Q'\right) \triangleq \frac{{\mathcal{\bar R}}^{\text{ZF}}(G,Q)}{{\mathcal{\bar R}}^{\text{ZF}}(1,Q')}   
    	= \xi \frac{G Q}{ Q'} \frac{\ln \left( 1+ \frac{P_t}{G} \left( \frac{L}{Q} -1 \right) \right)}{\ln \left( 1+ P_t \left( \frac{L}{Q'} -1 \right)\right)},\label{gain_ZF_eq}\\
    	&{\Gc}_{\rm RZF} \left(G,Q;1,Q'\right) \triangleq \frac{{\mathcal{\bar R}}^{\text{RZF}}(G,Q)}{{\mathcal{\bar R}}^{\text{RZF}}(1,Q')}   
 \  \stackrel{a.s.}{\longrightarrow}\  \xi \frac{\mathring{R}^{\rm RZF} (G, c L)}{\mathring{R}^{\rm RZF} (1, c' L)}, \label{gain_RZF_eq}
    \end{align}
    where $\mathring{R}^{\rm RZF} (\cdot, \cdot)$ is defined in~\eqref{Rate_RZF_eq}, and where $\xi \triangleq \frac{\left( L - Q\zeta_{G,Q} \right)}{\left( L - Q'\zeta_{1,Q'} \right)}$.
\end{coro}

	\vspace{-0.5cm}\section{Optimizing Physical Layer Vector Coded Caching}\label{L1_opt_sec}
	Theorems~\ref{Tight_Bound_MF_coro}-\ref{Rate_RZF_lemma} reveal the important dependence of vector coded caching on the number of streams, $Q$, that we choose to activate. This dependence strikes at the very core of the problems stemming from power-splitting and CSI overheads.   Indeed, while an increased $Q\leq L$ allows for a higher DoF at lower subpacketization, this increase in the number of streams may not be beneficial in practice as it entails less power per stream as well as more CSI to be communicated. 
	
	For this reason, we here proceed to analytically optimize our schemes over the choices of $Q$. This optimization is tractable partly due to the simplicity of the achievable-rate expressions derived in the previous theorems, and while some of these expressions involve asymptotic approximations, they will, as we will verify numerically, be very precise (see for example Fig.~\ref{Rate_Opt_fig}). 
	Our analysis of the optimal $c^*$ will assume a variable $c=Q/L$ that is continuous and unbounded. As noted before, the optimization takes into account the impact of CSI acquisition under TDD. 
	
    Let us first focus on deriving the optimal $c^*$ for MF precoding, where we consider $c\in(0,\infty)$ and $\Omega \triangleq \frac{P_t}{P_t+G}$.
    \begin{theorem}\label{MF_cost_lemma}
    		In the MF-based $(G,Q)$-vector coded caching with non-negligible CSI costs, the optimal $c^*$ that maximizes ${\mathcal{\bar R}}^{\text{MF}}$ in the asymptotic sense, is given by the solution to the following:
    		\begin{align}\label{MF_root_c}
    				(1-2 \zeta_{G,Q} c^*)  \ln \Big( 1+ \frac{\Omega}{c^*}\Big) - \frac{\Omega (1-\zeta_{G,Q} c^*)}{\Omega+c^*} =0.
    		\end{align}
    \end{theorem}
    \vspace{-0.3cm}
    \begin{proof}
    	We prove the theorem by demonstrating that ${\mathcal{\bar R}}^{\rm MF}$ as derived in Corollary~\ref{cor:EffectiveRates3} is concave over $c \in (0, \infty)$. Let us first note that 
    	the first derivative of~$\bar R^{\rm MF}$ in \eqref{R_MF_sum_caseii} is given by
    	\begin{align}\label{first_der_MF_sumrate}
    			\frac{\partial  \bar R^{\rm MF}}{\partial c} = G L \left[ \ln \left( \frac{\Omega + c}{c} \right)
    			+\frac{c}{\Omega +c} -1\right],
    	\end{align}
    	whereas the second derivative is then given by
        \begin{align}\label{R_MF_sec_der}
            \frac{\partial^2  \bar R^{\rm MF}}{\partial c^2}
            = -G L \frac{\Omega^2}{c(\Omega+c)^2} <0.
        \end{align}
        By differentiating ${\mathcal{\bar R}}_{\rm sum}^{\rm MF}$ in~\eqref{Rate_CSI_Def} with respect to $c$, we  have that
        \begin{align}
            \frac{\partial {\mathcal{\bar R}}^{\rm MF}}{\partial c}
            &= (1-\zeta_{G,Q} c)  \frac{\partial  \bar R^{\rm MF}}{\partial c} - \zeta_{G,Q}  \bar R^{\rm MF}, \qquad\qquad
             \frac{\partial^2 {\mathcal{\bar R}}^{\rm MF}}{\partial c^2}
            = (1- \zeta_{G,Q} c) \frac{\partial^2  \bar R^{\rm MF}}{\partial c^2} - 2 \zeta_{G,Q} \frac{\partial  \bar R^{\rm MF}}{\partial c}. \label{derivatives_wcost} 
        \end{align}    
        Let us know inspect the signs of these derivatives. 
        First, note that $\frac{\Omega + c}{c} \geq 1$ for any feasible $\Omega$, $c$, simply because $\Omega = \frac{P_t}{P_t+G}\geq 0$. Let us also note that the function $\ln(x)+1/x$ is decreasing when $x \in (0,1)$ and is increasing when $x \in [1,\infty)$, and also that its minimum value --- attained at $x=1$ --- is equal to~$1$.  
        Consequently, it follows that 
    	\begin{align}\label{R_MF_sum_zero}
    	    \frac{\partial \bar R^{\rm MF}}{\partial c} = G L \left[ \ln \left( \frac{\Omega + c}{c}\right)
    			+\frac{c}{\Omega + c} -1\right]   \ge 0,
    	\end{align}
    	where the inequality is strict unless $ \frac{\Omega + c}{c}=1$ corresponding to $c \to \infty$. Therefore, we  conclude that $\bar R^{\rm MF}$ is monotonically increasing over $c \in (0,\infty)$.
    	
        From the fact that $\frac{\partial  \bar R^{\rm MF}}{\partial c} \ge 0$ (cf.~\eqref{R_MF_sum_zero}), the fact that $\frac{\partial^2  \bar R^{\rm MF}}{\partial c^2}  <0$ (cf.~\eqref{R_MF_sec_der}), and the fact that $1-\zeta_{G,Q} c \ge 0$, we can conclude that $\frac{\partial^2 {\mathcal{\bar R}}^{\rm MF}}{\partial c^2} <0$ in \eqref{derivatives_wcost}, and therefore ${\mathcal{\bar R}}^{\rm MF}$ is concave over $c \in (0, \infty)$, and thus that the global maximum point of ${\mathcal{\bar R}}^{\rm MF}$ is at the root $c^*$ of $\frac{\partial {\mathcal{\bar R}}^{\rm MF}}{\partial c}$. 
    \end{proof}
Next, we consider ZF-based cache-aided precoding, for which we have the following.\vspace{-0.3cm}
	\begin{theorem}\label{ZF_cost_lemma}
			In the ZF-based $(G,Q)$-vector coded caching with non-negligible CSI costs, the optimal $c^*$ that maximizes ${\mathcal{\bar R}}^{\text{ZF}}$, is given by the solution to the following equation:
			\begin{align}\label{zero_ZF_cost_identi}
			 \big( 1- 2 \zeta_{G,Q} c^* \big)   \ln \left( 1+ \frac{P_t}{G} \left( \frac{1}{c^*}-1 \right)\right) -  \frac{ (1- \zeta_{G,Q} c^*)P_t / G}{(1-P_t/G)c^*+P_t/G}=0.
			\end{align}
	\end{theorem}
    \begin{proof}
        The proof builds on the properties of the first and second derivatives of ${\mathcal{\bar R}}^{\rm ZF}$, in a similar manner as in the proof of Theorem~\ref{MF_cost_lemma}. These derivatives now take the form 
        $\frac{\partial {\mathcal{\bar R}}^{\rm ZF}}{\partial c} = (1-\zeta_{G,Q} c)  \frac{\partial  \bar R^{\rm ZF}}{\partial c} - \zeta_{G,Q}  \bar R^{\rm ZF}$, and 
        $\frac{\partial^2 {\mathcal{\bar R}}^{\rm ZF}}{\partial c^2} = (1- \zeta_{G,Q} c) \frac{\partial^2  \bar R^{\rm ZF}}{\partial c^2} - 2 \zeta_{G,Q} \frac{\partial  \bar R^{\rm ZF}}{\partial c}$.  
After applying~\eqref{ZF_sumrate}, these derivatives take the form
		\begin{align}
			 &\frac{\partial \bar R^{\rm ZF}}{\partial c} = G L \left[ \ln \left( 1+ \frac{P_t}{G} \left( \frac{1}{c}-1 \right)\right) - \frac{P_t / G}{(1-P_t/G)c+P_t/G}\right], \label{first_der_ZF} \\
			 &\frac{\partial^2 \bar R^{\rm ZF}}{\partial c^2} = - G L \frac{(P_t/G)^2}{c \big( \left(1-P_t/G \right) c +P_t/G \big)^2} <0.\label{second_der_ZF}
		\end{align}
		Since the second derivative $\frac{\partial^2 \bar R^{\rm ZF}}{\partial c^2}$ in~\eqref{second_der_ZF} is always negative, $\bar R^{\rm ZF}$ is a concave function with respect to $c$. 
		Therefore, the root of $\frac{\partial \bar R^{\rm ZF}}{\partial c}=0$, which we denote by $c^\star_R$, is the global maximum of $\bar R^{\rm ZF}$ over $c \in (0,\infty)$. 
		Moreover, it follows from~\eqref{ZF_sumrate} that $\bar R^{\rm ZF}=0$ for $c=1$ and that $\bar R^{\rm ZF}>0$ for $0<c<1$, which implies that $c^\star_R$ belongs in the interval $(0,1)$. 
		
        Since $\frac{\partial \bar R^{\rm ZF}}{\partial c} \big|_{c=c^\star_R}=0$ and since $\frac{\partial^2 \bar R^{\rm ZF}}{\partial c^2}$ is always negative, 
        we know that $\frac{\partial  \bar R^{\rm ZF}}{\partial c}$ is monotonically decreasing and that this same $\frac{\partial \bar R^{\rm ZF}}{\partial c}$ is negative for all $c \in (c^\star_R,1)$. 

        Consequently, ${\mathcal{\bar R}}^{\rm ZF}$ is monotonically decreasing in the interval $c \in (c^\star_R,1)$. Thus the maximum point of ${\mathcal{\bar R}}^{\rm ZF}$ must belong in the interval $(0,c^\star_R)$ where we can see that $\frac{\partial \bar R^{\rm ZF}}{\partial c} >0$ and $\frac{\partial^2  \bar R^{\rm ZF}}{\partial c^2} <0$. 
        Hence, ${\mathcal{\bar R}}^{\rm ZF}$ is concave throughout $c \in (0,c^\star_R)$, and thus the root of $\frac{\partial {\mathcal{\bar R}}^{\rm ZF}}{\partial c}$ is the global maximum point of ${\mathcal{\bar R}}^{\rm ZF}$, where this point $c^*$ must belong in $(0,c^\star_R)$. Finally, substituting \eqref{ZF_sumrate} and \eqref{first_der_ZF} into $\frac{\partial {\mathcal{\bar R}}^{\rm ZF}}{\partial c}$ yields \eqref{zero_ZF_cost_identi} and proves the theorem. 
    \end{proof}

	\begin{remark}\label{high_SNR_opt_ZF_remark}
	As $P_t \to \infty$, we can write \eqref{first_der_ZF} as 
		$\frac{\partial \bar R^{\rm ZF}}{\partial c}  = G L  \left[ \ln \left( \frac{P_t}{G} \right) + \ln \left( \frac{1-c}{c} \right)  - \frac{1}{1-c} \right] + o(1),$
		where $\lim_{P_t \to \infty} o(1)=0$.
	    Therefore, in the high-SNR regime and without taking CSI costs into account, the optimal value of $c$ that maximizes ${\mathcal{\bar R}}^{\text{ZF}}$ and thus\footnote{Recall that in the high SNR regime, ZF and RZF coincide.} ${\mathcal{\bar R}}^{\text{RZF}}$, is given by 
	    $c^* = \left( 1+\frac{1}{\mathcal{W} \left(  P_t/(eG) \right)} \right)^{-1}$,
	    upon omitting an $o(1)$ additive term, and upon using $\mathcal{W}(\cdot)$ to denote the Lambert W-Function. This expression can serve as a good approximation in those moderate-to-high SNR scenarios where the dimensionality of the problem implies a relatively small CSI cost. As one can see, as the SNR becomes very large, the above $c^*$ converges, as is known, to 1, corresponding to $Q\approx L$. 
	\end{remark}
	\vspace{-0.3cm}

		Having derived the above optimal $c^*$, we can now consider the ratio
		\begin{align} \label{Ratio1}
		    \Gc^\star \triangleq \frac{\max_{Q \in \mathbb{Z}^+} {\mathcal{\bar R}^i}(G, Q)}{\max_{Q' \in \mathbb{Z}^+} {\mathcal{\bar R}^i} (G=1,Q')},
		\end{align}
		which describes the performance boost due to caching, over (independently) optimized downlink cacheless systems, after accounting for CSI costs. These gains $\Gc_{\text{MF}}^\star, \Gc_{\text{ZF}}^\star, \Gc_{\text{RZF}}^\star$ are reported for the three precoders of interest. As one would expect, this comparison is done under a fixed set of SNR and antenna resources. The transition from the continuous $c$ to the operating $Q$, will follow by simply considering $Q^* = \argmax_{ Q \in \{ \lfloor c^*L  \rfloor, \lfloor c^*L  \rfloor+1 \}} \big\{  \mathcal{\bar R} \big( Q \big) \big\}$, where $\lfloor \cdot \rfloor$ denotes the nearest integer less than or equal to the argument.

\vspace{-0.3cm}\section{Numerical Results}\label{numerical_sec}
We proceed to numerically demonstrate the achieved effective rates as well as the effective gains that an optimized vector coded caching scheme provides over the independently optimized cacheless downlink solution. We note that the simulated results employ no approximations\footnote{For example, the corresponding SINR is taken directly from~\eqref{SINR_actual_def}.}. 
For ease of exposition, we list in Table~\ref{thm_list} the derived theorems and corollaries. 
%

\renewcommand*{\arraystretch}{1.35}%
\begin{table}[t]
  \centering~\vspace{-0.5cm}
  \caption{Derived Theorems (Thms.) and Corollaries (Cors.)}\label{thm_list}
  \label{Scenario_Tab}
   \scalebox{0.844}{
  \begin{tabular}{| c| c | c | c | c | c | c | c |}
  \hline
  Thm. \ref{Tight_Bound_MF_coro}  & Thm. \ref{Rate_ZF_lemma} & Thm. \ref{Rate_RZF_lemma} 
  & Thm. \ref{MF_cost_lemma} & Thm. \ref{ZF_cost_lemma} & Cor. \ref{cor:EffectiveRatesCacheless} & Cor. \ref{cor:EffectiveRates3}  & Cor. \ref{cor:gains} \\
  \hline  		
  \makecell[c]{Average \\ sum-rate MF} & \makecell[c]{Average \\ sum-rate ZF} & \makecell[c]{Average \\ sum-rate RZF}  & \makecell[c]{Optimal $Q$ \\ for MF }   &\makecell[c]{Optimal $Q$ \\ for ZF} &\makecell[c]{Average sum-rate\\ in cacheless MF} & \makecell[c]{Effective rates\\in MF/ZF/RZF}  & \makecell[c]{Effective gains\\in MF/ZF/RZF}\\
  \hline     	
  \end{tabular}~\vspace{-3ex}
   }
\end{table}
\renewcommand*{\arraystretch}{1}%

    \vspace{-0.5cm}
    \subsection{Multiplicative Boost of Vector Coded Caching over Downlink Systems}
    The following figures build on the analysis of the effective sum rates and effective gains of Section~\ref{large_ana_sec}, as well as on the analysis of the optimized gains of Section~\ref{L1_opt_sec}. These figures incorporate the CSI costs in the realistic scenario of having $\beta_{\rm tot} = 10$, $T_c = 0.04$ seconds and $W_c = 300$~kHz (cf.~\eqref{Rate_CSI_Def}). 
        Figure~\ref{Rate_Opt_fig} (left) describes the effective rate of the different cache-aided schemes, for different values of $Q$.  The plot highlights the tightness of the results of Theorems \ref{Tight_Bound_MF_coro}-\ref{Rate_RZF_lemma} (after accounting for CSI costs: see Corollary~\ref{cor:EffectiveRates3}), where we see that indeed the derived asymptotically-approximate expressions have no discernible distance from the actual (simulated) performance. The vertical
lines indicate the optimal $Q$ derived in Theorems~\ref{MF_cost_lemma}-\ref{ZF_cost_lemma}. These optimal points indeed match the actual maximum point of the curves. 
Figure~\ref{Rate_Opt_fig} (right) extends this illustration of the tightness of the results, to Theorems \ref{MF_cost_lemma}-\ref{ZF_cost_lemma}, by illustrating the optimized (over all $Q$ choices) effective rate performance of the three precoders, comparing the derived results\footnote{Here we note that for the case of RZF precoding, the figure plots the result of Theorem~\ref{Rate_RZF_lemma} (Corollary~\ref{cor:EffectiveRates3}), considering a $c^*$ value that is drawn from an exhaustive search based on these derived expressions.}, to the actual performance.


		\begin{figure}[t]\vspace{-0.1cm}\centering
				\begin{subfigure}[t]{.499\textwidth}\centering		
				\includegraphics[width=3.2 in]{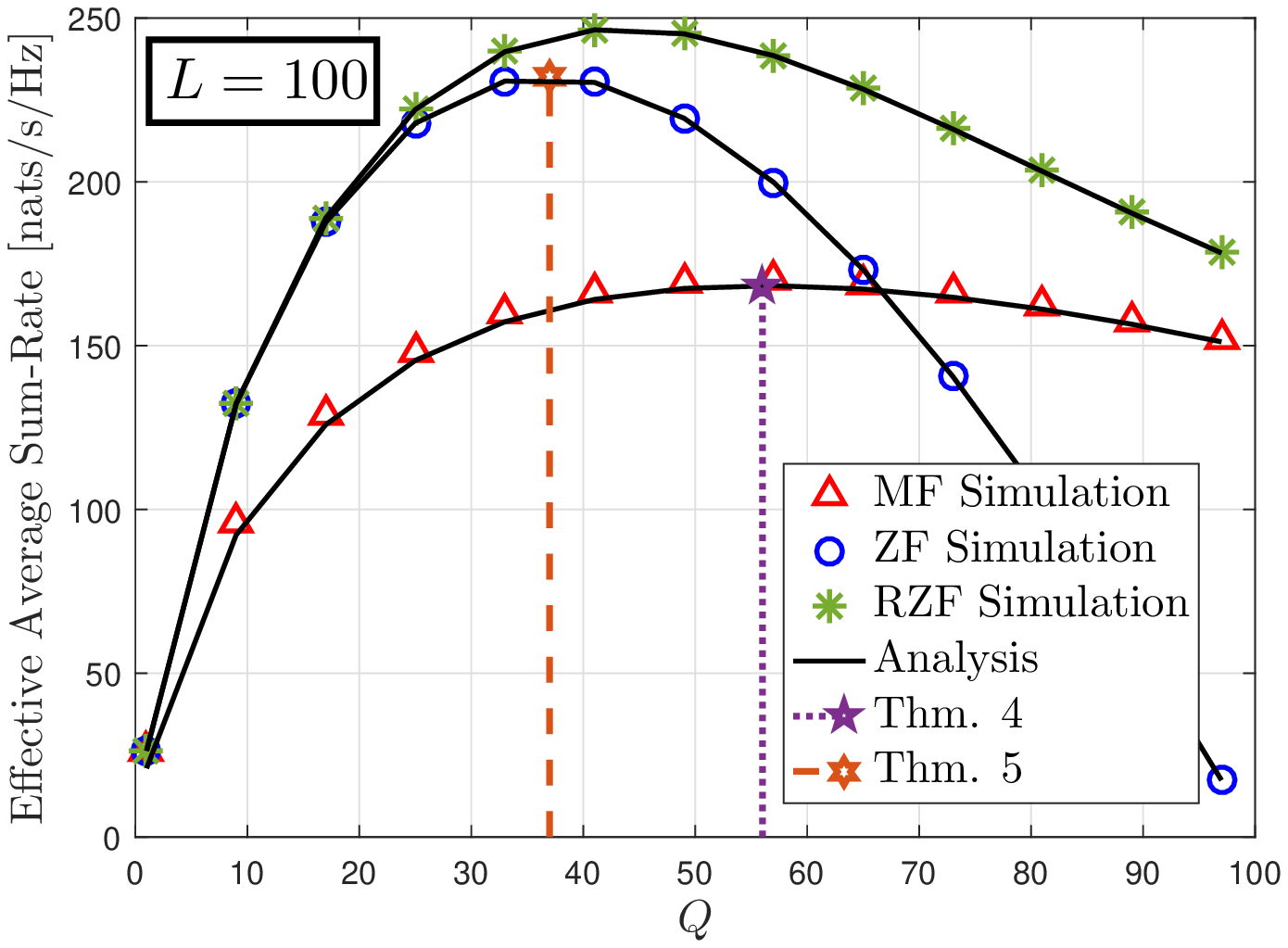}
				\end{subfigure}~%
				\begin{subfigure}[t]{.499\textwidth}\centering					
				\includegraphics[width=3.2 in]{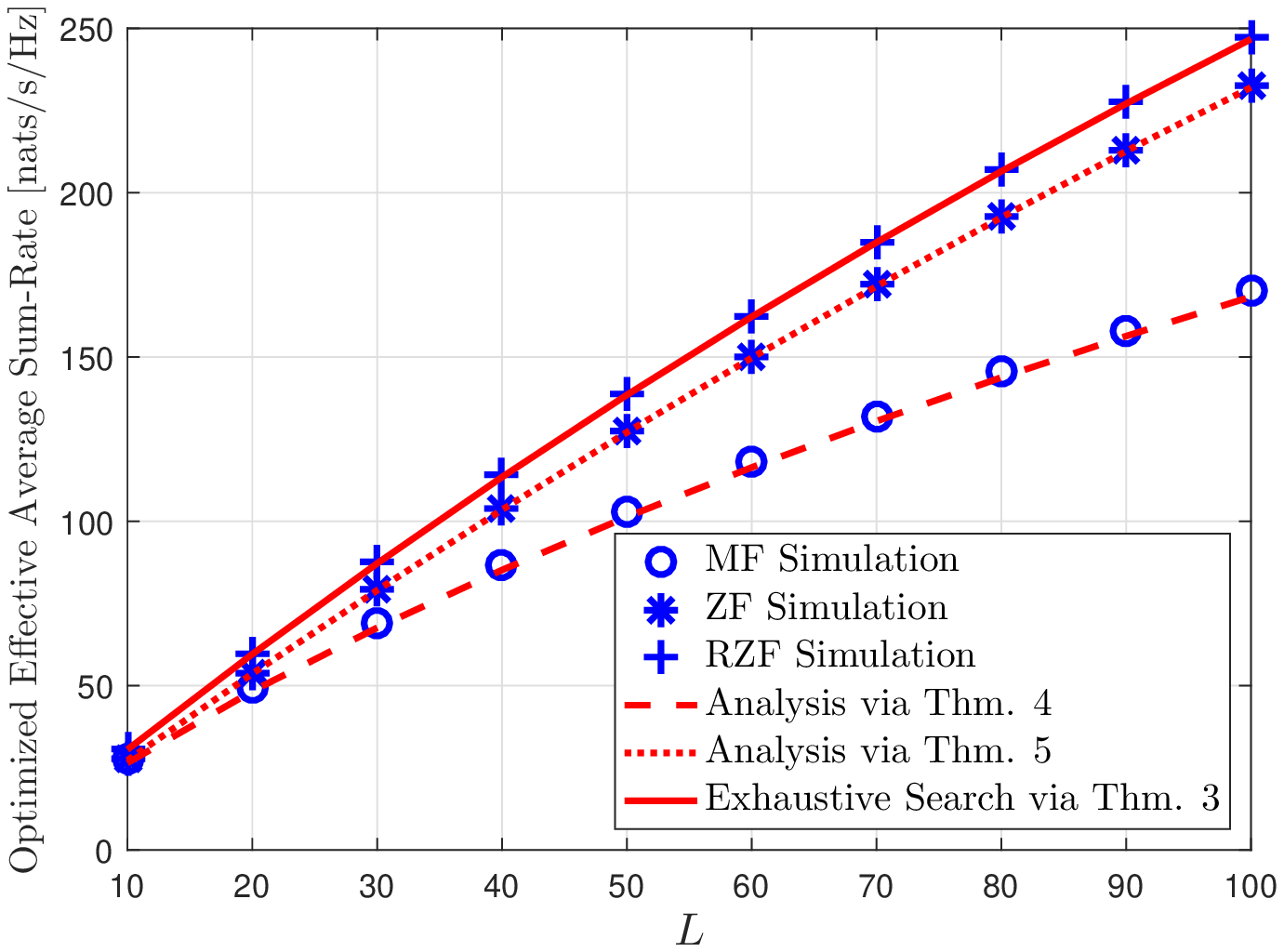}
				\end{subfigure}~\vspace{-0.5cm}
			\caption{Effective rate $\mathcal{\bar R}$ and optimized effective rate for $P_t =10$ dB and $G=5$.}\vspace{-0.4cm}\label{Rate_Opt_fig}
		\end{figure}		
		\begin{figure}[t]\vspace{-0.5cm}\centering
				\begin{subfigure}[t]{.499\textwidth}\centering		
				\includegraphics[width=3.2 in]{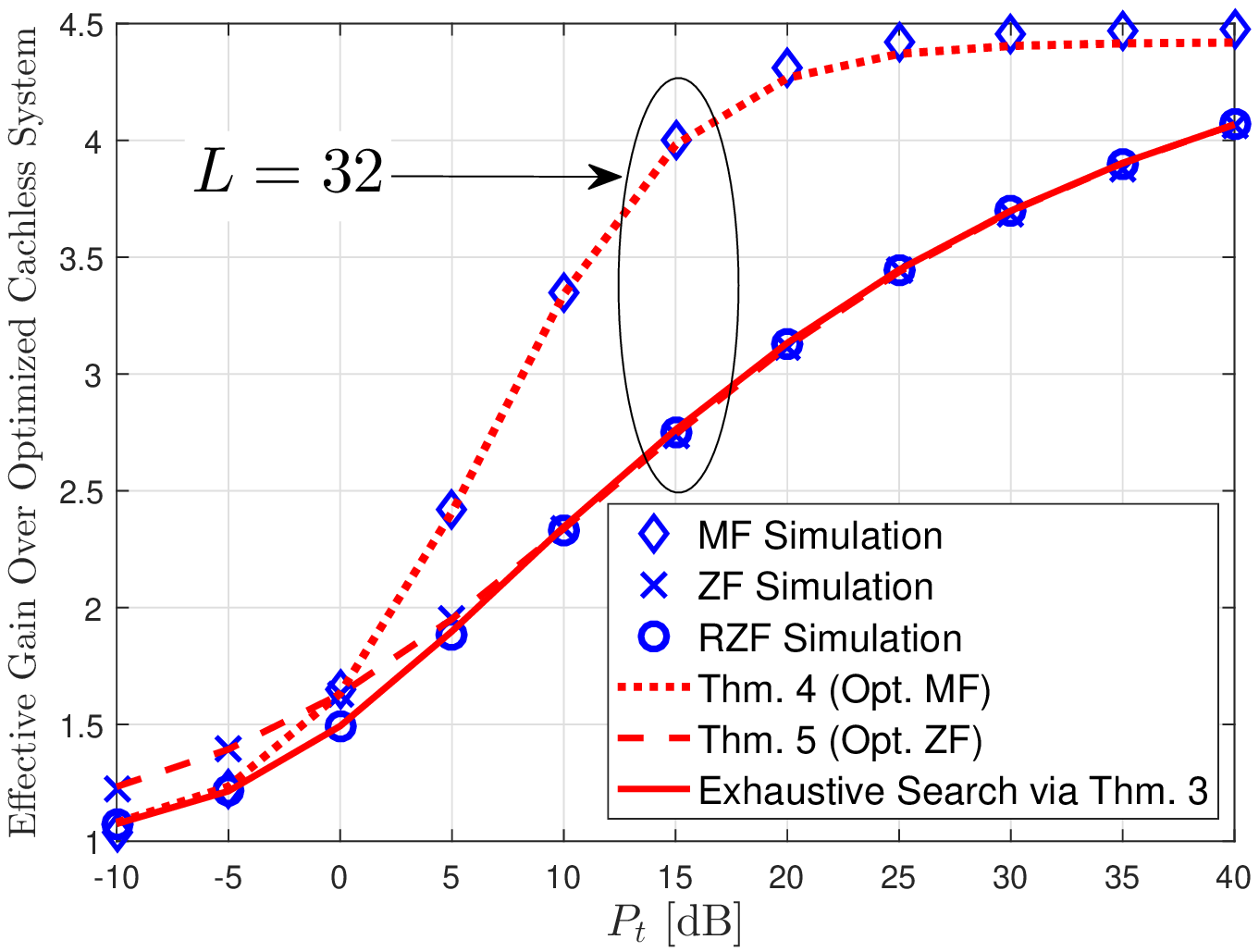}
				\end{subfigure}~%
				\begin{subfigure}[t]{.499\textwidth}\centering					
				\includegraphics[width=3.2 in]{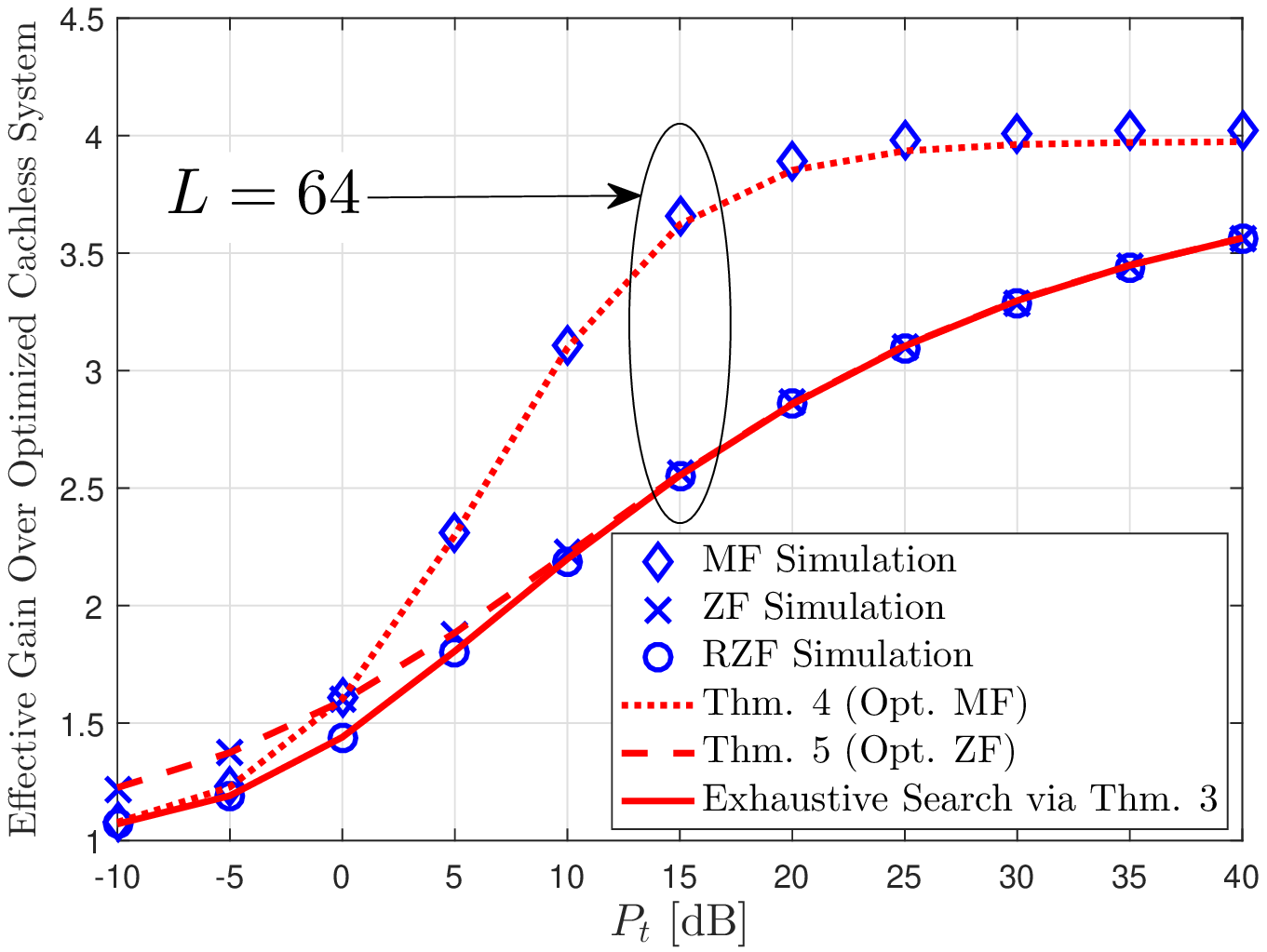}
				\end{subfigure}~\vspace{-0.5cm}
			\caption{Effective gain $\Gc^\star$ over optimized cacheless system for $L\in \{32, 64\}$ and $G=6$.}\vspace{-.3cm}\label{Actual_Boost_fig}
		\end{figure}

		Fig.~\ref{Actual_Boost_fig} focuses on the effective gains over optimized cacheless downlink systems. As before, the theoretical and simulated results match fully. Here the theoretical results reflect the effective gain ratio $\Gc^\star$ in \eqref{Ratio1}, where the derived effective-rate expressions are from Corollary~\ref{cor:EffectiveRates3} (and the corresponding Theorems~\ref{Tight_Bound_MF_coro}-\ref{Rate_RZF_lemma}), and where the optimizing $c^*$ are directly from\footnote{We recall that in the RZF case, the figure plots Theorem~\ref{Rate_RZF_lemma} and Corollary~\ref{cor:EffectiveRates3}, after a numerically evaluated $c^*$.} Theorems~\ref{MF_cost_lemma}-\ref{ZF_cost_lemma}.

		Under the above realistic coherence periods and coherence bandwidths, realistic CSI costs, as well as realistic values of SNR and $L$, the multiplicative boosts over the achievable rates of optimized downlink systems are quite notable. For example, for $64$ transmit antennas, a receiver-side SNR of $20$ dB, the same $W_c = 300$ kHz and $T_c=40$ ms, and under realistic file-size and cache-size constraints that allow us to assume $G=6$, vector coded caching is here shown to offer a multiplicative boost of about $280\%$ in ZF/RZF precoding and $380\%$ over MF-based cacheless systems, whereas for the case of $32$ antennas the gain elevates to $310\%$ for ZF and to a $430\%$ multiplicative boost in the performance of already optimized MF-based cacheless systems\footnote{In addition to the speedup factor reported here, the use of caches can also lead to additional --- albeit marginal --- reductions in delivery-time, complements of the so-called local caching gain, which is though of no particular interest to this study.}. As one would expect, this same figure reveals that the gains $\Gc^\star$ grow monotonically with the SNR, and often come very close to the theoretical upper bound of $G$.

Another interesting comparison is shown in Fig.~\ref{Gain_fig}, where we ask that the cache-aided and cacheless scenarios share the same exact multiplexing gain $Q$. The motivation for this comparison traces back to the idea of channel hardening, which refers to the fact that as long as $L$ is sufficiently large, and as long as $Q/L$ is sufficiently small, the channel converges to a deterministic value, thus making CSI acquisition easier. While this paper is not about the channel hardening properties of the cache-aided downlink, this Fig.~\ref{Gain_fig} --- which plots the effective gain ${\Gc}(G,Q;1,Q) ={\mathcal{\bar R}}(G,Q) / {\mathcal{\bar R}}(1,Q)$ --- offers a first indication of yet another benefit of vector coded caching, which now allows us to serve more users at a time, but do so with a controlled ratio $Q/L$ that guarantees certain channel hardening conditions. Focusing on the case of a fixed $Q=8$ for both the cache-aided ($G=6$), as well as the cacheless case $(G=1)$, Fig.~\ref{Gain_fig} reveals that under the same $W_c, T_c$ and under realistic SNR values of, for example, approximately 15dB, the effective gains (over cacheless equivalent systems with the same $Q/L$) approach $400$\% for the ZF-based precoders, and even go beyond $540$\% when using MF-based precoding. 
Similar gains are recorded in the larger scenario with $L=128$ transmit antennas.

			\begin{figure}[t]\vspace{-0.4cm}\centering
					\begin{subfigure}[t]{.499\textwidth}\centering		
					\includegraphics[width=3.2 in]{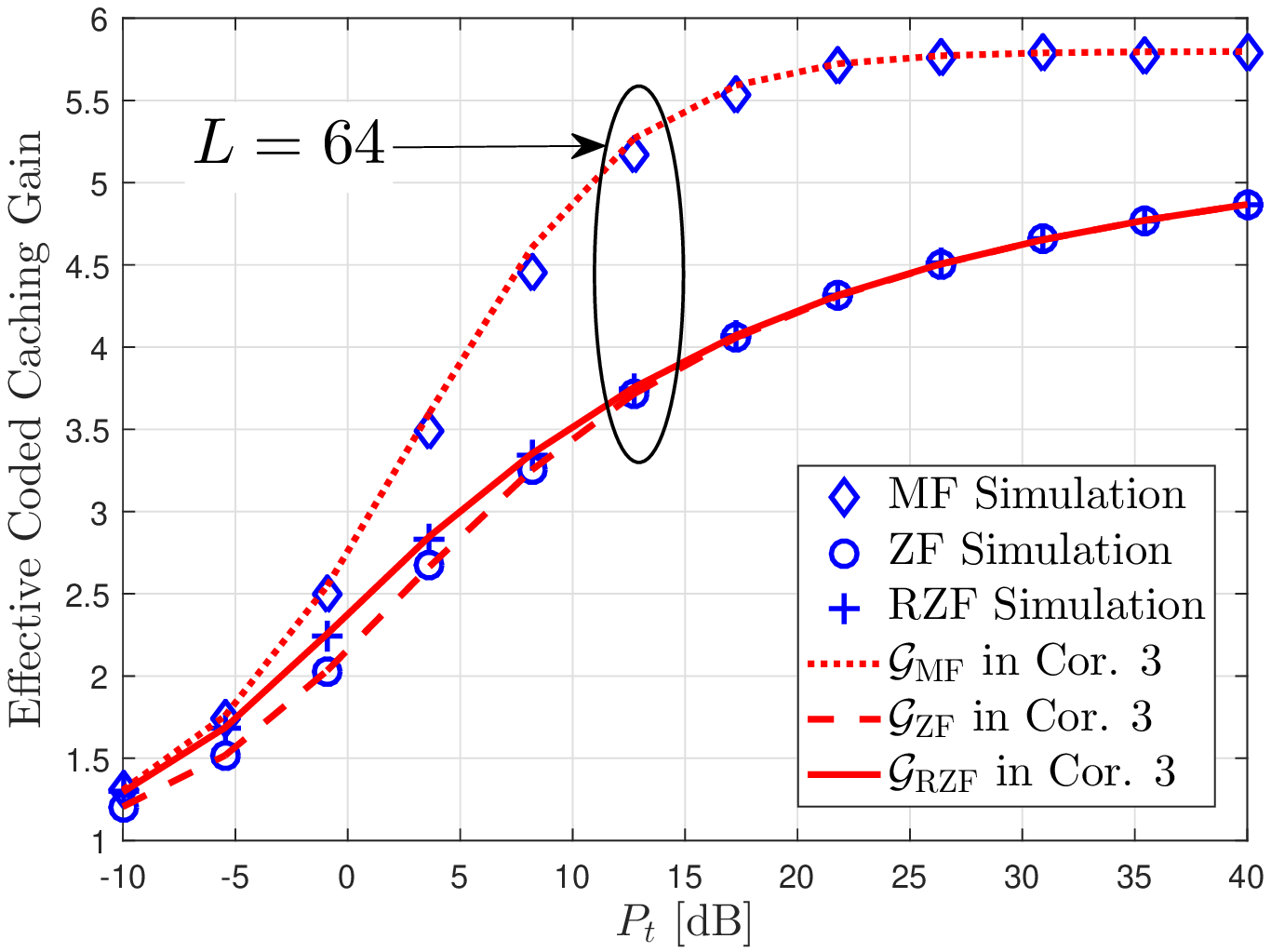}
					\end{subfigure}~%
					\begin{subfigure}[t]{.499\textwidth}\centering					
					\includegraphics[width=3.2 in]{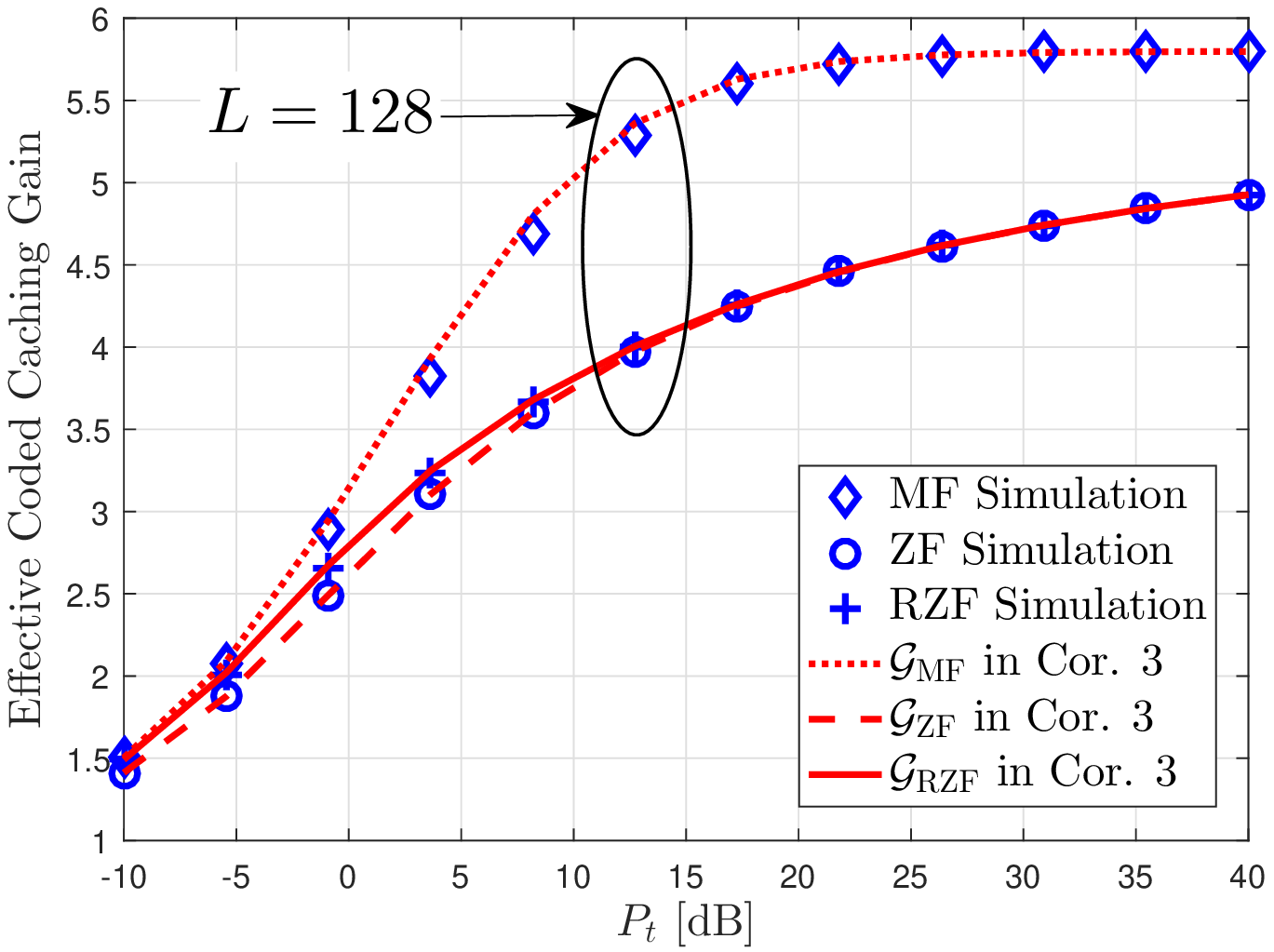}
					\end{subfigure}~\vspace{-0.3cm}
				\caption{Hardening-constrained effective gain over a constrained classical downlink system. $Q$ is fixed for both systems at $Q=8$, while $G=6$.}\vspace{-0.3cm}\label{Gain_fig}
			\end{figure}
The above numerical illustrations refer to theoretical gains of $G=5$ and $G=6$. To give the reader a sense of what such values entail, we offer the following simplifying example scenario. 
\begin{example}
Let us consider the Netflix library, focusing on movies, and let us make the educated speculation that the library entails a zipf parameter close to $1.4$ (see for example~\cite{7945355}). Assume that we choose to apply coded caching on the part of the library that captures $90$\% of the traffic, such that on average, $90$\% of the Netflix traffic will experience a streaming volume reduction by a (theoretical) factor of $G$. Let us assume that the receiving devices are each endowed with a cache of size equal to $25$GB, and let us assume that they stream HD movies whose size is approximately $1.3$GB. The subpacketization constraint will be largely defined by the latency requirements. Assume a latency of two minutes, which can be seamlessly handled with a small buffer. Assuming movies that have about $90$ minutes duration, this translates to file (sub-movie) sizes of approximately 28.8MB. Under the assumption of atomic communication packets of size equal to approximately $50$ bytes, this brings us to a subpacketization of $6 \cdot 10^5$, which allows for a theoretical gain of $G=7$. Under approximately the same conditions, but for Full-HD movies of size $2.47$GB, the corresponding $\gamma< 1/10$, can allow for a gain close to $G=6$. Assuming individual cache sizes of $5$GB, and Standard Definition (SD-$480$p) streaming, the file (sub-movie) sizes become 8.9MB, and thus, with an atomic communication packet size of $200$ bytes, we have subpacketization $4.5 \cdot 10^4$, and a theoretical gain of $G=5$. Going back to our example of the hardening-constrained setting with $Q = 8$, and under the Full-HD assumption, we see that to attain the promised gain of $G=6$ requires a network with at least $Q\Lambda \approx 400$ receiving nodes/antennas, which could represent $K=100$ users with $4$ receive antennas each. Similarly in the aforementioned HD scenario, attaining the theoretical gain of $G=7$ would entail at least $Q\Lambda \approx 240$ receiving nodes/antennas, which could represent $60$ users with $4$ receive antennas each, while in the SD small-cache scenario this corresponds to $34$ single antenna users, or $17$ users with $2$ antennas each. 
\end{example}			

\FloatBarrier

\section{Conclusions}\label{conclude_sec}

This work explores new methods for improving the performance of advanced multi-antenna downlink systems. Such systems constitute the backbone of modern wireless communications, and they have traditionally depended on an optimized interplay between multiplexing and beamforming techniques. While multi-antenna arrays have been without a doubt a most valuable resource and the driving force behind advanced communications technologies, we are now presented with a new and highly complementary and abundant resource in the form of the ever-increasing storage volumes available across even the smallest of communicating nodes.

Motivated by the opportunity offered by this newly abundant resource, our work presented very simple to implement optimized cache-aided linear precoding schemes for the multi-antenna downlink broadcast channel. These schemes simply exploit cached content in order to be able to simultaneously transmit carefully selected precoding vectors that would have otherwise been sent one after the other. Because of the simplicity of this idea, it is conceivable to expect the gains to persist for a broader class of precoders. 
Our performance analysis derives simple expressions that reveal 
significant multiplicative gains from applying caching over already optimized downlink systems, where these gains persist for various well-known precoding classes. This same analysis and optimization are here shown to hold very tight in realistic non-asymptotic settings, while also incorporating a variety of practical considerations such as power dissemination across signals, realistic SNR values, as well as CSI costs. The comparisons of optimized cache-aided vs. optimized cacheless downlink systems, reveal that vector coded caching can recover a sizeable portion of its theoretic (high-SNR) gain $G = \Lambda\gamma+1$, even in realistic wireless settings operating at realistic SNR values. 


In terms of challenges, indeed $G$ remains, under current practices, bounded in the range of single digits. Any improvement beyond this range, would require either a dramatic increase in the storage capability of nodes ($\gamma$), or a dramatic research breakthrough in the area of subpacketization-constrained coded caching. Further improving the subpacketization-constrained performance of coded caching primitives (thus effectively allowing for a larger $\Lambda$), remains to date the big challenge in coded caching, and any progress in that direction would undoubtedly have a profound impact on the performance of cache-aided multi-antenna systems.

The reported gains here will naturally come under pressure from additional realistic considerations such as having statistically asymmetric channels, which is a problem though that can be partially ameliorated with power control, with rate-splitting approaches \cite{LampirisJJPEOsvaldo} \cite{HamdiTiT2021}, or with the novel `brothers' approach in~\cite{Zhao2020_TWC}. 
These same reported gains may also come under pressure from the additional CSI costs that would arise in the event where multi-antenna coded caching algorithms start serving more and more users. Remedies for this can be found in the novel clique structures recently reported in~\cite{lampiris2021_Bottleneck}. 
A big associated open problem is the simultaneous reduction of both the subpacketization and CSI costs (see~\cite{LampirisCSITandSubpackSPAWC} for some early efforts). Naturally the system performance also remains subject to the need for cacheable and live-streamed data to co-exist (see~\cite{HamdiTiT2021}), the need for cache-aided and cacheless users to coexist\footnote{See~\cite{LamEliCachelessTit}, which reveals the surprising conclusion that cacheless users can benefit from full coded caching gains.}, as well as will depend on the stochastic nature of the network topology and user behavior (for some early remedies, the reader can refer to \cite{AdeelTCOM21,AdeelTON}).

The presented new results as well as the aforementioned challenges, arrive at an instance when bandwidth and antenna resources are asked to handle an aggressively increasing volume of data. At the same time though, the new results come at a time when Moore's law on storage capabilities remains intact, as well as at a time when the ever-increasing majority of communicated content is cacheable~\cite{Ciso_forest}. For these reasons, and given the powerful gains reported here, we believe that the aforementioned techniques can further help translate the abundance of Gbytes of storage space, into much needed spectral efficiency.

\vspace{-0.2cm}
\appendices
\renewcommand{\thesectiondis}[2]{\Roman{section}:}
\section{Proof of Theorem \ref{Tight_Bound_MF_coro}}\label{Proof_Thm1_Zhang}
	Similar to the proof of \cite[Lemma 1]{Zhang}, we define $X \triangleq  \frac{P_t}{G c L^2}   \big| {\bf h}_{\psi,k}^T {\bf h}_{\psi,k}^* \big|^2$ and $Y \triangleq 1 + \frac{1}{Q} \sum_{\vartheta=1, \vartheta \neq k}^{Q} Y_\vartheta$, where $Y_\vartheta \triangleq  \frac{P_t}{G L} \big| {\bf h}_{\psi,k}^T {\bf h}_{\psi,\vartheta}^* \big|^2$. From \cite[Lemma 1]{Yeon_Geun_Lim}, we know that $\mathbb{E}\{X\} = \frac{P_t}{c G}(1+1/L)$, ${\rm Var}\{X\} = \frac{P_t^2}{G^2c^2} \Big(\frac{4}{L}  +\frac{10}{L^2} + \frac{6}{L^3}\Big) <\infty$, $\mathbb{E}\{Y_\vartheta\} = P_t/G$ and ${\rm Var}\{Y_\vartheta\} = \frac{P_t^2}{G^2}(1+2/L) < \infty.$ We want to prove that
	\begin{align}\label{identity_proof}
	    \frac{\bar{R}^{\rm MF}(G,cL)}{c \ GL} = \mathbb{E}\left\{ \ln \left( 1+ \frac{X}{Y} \right) \right\} = \ln \left( 1 +  \frac{\mathbb{E}\{X\}}{\mathbb{E}\{Y\}} \right) + o(1), \text{ as } Q=cL \to \infty.
	\end{align}
	
	By applying Jensen’s inequality on $\mathbb{E}\left\{ \ln \left( X+Y \right) \right\}$ and $\mathbb{E}\left\{ \ln \left( Y \right) \right\}$ separately, we can get the following bounds:
	\begin{align}
	   &\ln\left( \frac{1}{  \mathbb{E}\left\{ (X+Y)^{-1} \right\} }  \right) \le \mathbb{E}\left\{ \ln \left( X+Y \right) \right\} \le \ln \left( \mathbb{E}\left\{ X+Y \right\} \right) \\
	   & -\ln \left( \mathbb{E}\left\{ Y \right\} \right)   \le - \mathbb{E}\left\{ \ln \left( Y \right) \right\} \le - \ln \left( \frac{1}{ \mathbb{E} \{Y^{-1}\}} \right),
	\end{align}
	and after combining these two bounds, we get
	\begin{align}\label{bounds_ul}
	    \ln\left( \frac{1}{  \mathbb{E}\left\{ (X+Y)^{-1} \right\} }  \right) - \ln \left( \mathbb{E}\left\{ Y \right\} \right) \le \mathbb{E}\left\{ \ln \left( 1+ \frac{X}{Y} \right) \right\} \le \ln \left( \mathbb{E}\left\{ X+Y \right\} \right) - \ln \left( \frac{1}{ \mathbb{E} \{Y^{-1}\}} \right).
	\end{align}
	On the other hand, Jensen's inequality says that $\mathbb{E}\{Y^{-1}\} \ge 1/\mathbb{E}\{Y\}$ and $\mathbb{E}\left\{ (X+Y)^{-1} \right\} \ge 1/ \mathbb{E}\left\{ (X+Y) \right\}$, which yields
	\begin{align}
	   & \ln \left( 1 +  \frac{\mathbb{E}\{X\}}{\mathbb{E}\{Y\}} \right)  = \ln \left( \mathbb{E}\left\{ X+Y \right\} \right) - \ln \left(  \mathbb{E} \{Y\} \right)  \le \ln \left( \mathbb{E}\left\{ X+Y \right\} \right) - \ln \left( \frac{1}{ \mathbb{E} \{Y^{-1}\}} \right), \label{bouds_u}\\
	   & \ln \left( 1 +  \frac{\mathbb{E}\{X\}}{\mathbb{E}\{Y\}} \right)  \!=\! \ln \left( \mathbb{E}\left\{ X+Y \right\} \right) - \ln \left(  \mathbb{E} \{Y\} \right) \ge \ln\left( \frac{1}{  \mathbb{E}\left\{ (X+Y)^{-1} \right\} }  \right) - \ln \left( \mathbb{E}\left\{ Y \right\} \right). \label{bounds_l}
	\end{align}
		At this point, both $ \mathbb{E}\left\{ \ln \left( 1+ \frac{X}{Y} \right) \right\}$ and $\ln \left( 1 +  \frac{\mathbb{E}\{X\}}{\mathbb{E}\{Y\}} \right)$ are bounded above and below by the same bounds \eqref{bounds_ul}--\eqref{bounds_l}.  The gap between these bounds takes the form
	\begin{align}
	    \Delta & \triangleq  \left\{ \ln \left( \mathbb{E}\left\{ X+Y \right\} \right) - \ln \left( \frac{1}{ \mathbb{E} \{Y^{-1}\}} \right) \right\} - \left\{ \ln\left( \frac{1}{  \mathbb{E}\left\{ (X+Y)^{-1} \right\} }  \right) - \ln \left( \mathbb{E}\left\{ Y \right\} \right) \right\} \notag\\
	    & = \ln \left[ \Big( \mathbb{E}\left\{ X+Y \right\} \mathbb{E}\left\{ (X+Y)^{-1} \right\} \Big) \Big( \mathbb{E}\left\{ Y \right\} \mathbb{E}\left\{ Y^{-1} \right\} \Big) \right].
	\end{align}
	We want to show that this gap vanishes as $Q=cL \to \infty$. By expanding the Taylor series of $Y^{-1}$ at $\mathbb{E}\{Y\}$, we can have that
	\begin{align}\label{limit1_Thm1_proof}
	    \lim_{Q \to \infty} \mathbb{E}\left\{ Y \right\} \mathbb{E}\left\{ Y^{-1} \right\} 
	     &= \lim_{Q \to \infty} \mathbb{E}\left\{ Y \right\} \mathbb{E}\left\{ \frac{1}{\mathbb{E}\{Y\}} - \frac{(Y - \mathbb{E}\{Y\})}{\mathbb{E}^2\{Y\}}  +  \frac{(Y - \mathbb{E}\{Y\})^2}{\mathbb{E}^3\{Y\}}  + \cdots  \right\} \notag\\
	    & = 1  + \lim_{Q \to \infty}\mathbb{E}\{g(Y)\} \overset{(a)}{=} 1  + \mathbb{E}\Big\{\lim_{Q \to \infty} g(Y) \Big\} \overset{(b)}{=} 1,
	\end{align}
	where $g(Y) \triangleq \sum\limits_{n=2}^{\infty} (-1)^n \frac{(Y-\mathbb{E}\{Y\})^n}{\mathbb{E}^n \{Y\}}$, where
	$(a)$ follows from exchanging the order of the limitation and expectation operators (validated via the Dominated Convergence Theorem (DCT))\footnote{To see this, first define $Z \triangleq |Y-\mathbb{E}\{Y\}| \ge 0$. As $Q \to \infty$, $Z \to 0$ (due to the law of large numbers), there always exists a constant $Q_0$ and $\varepsilon <1$  such that $Z < \varepsilon$ for any $Q>Q_0$. For $Z < \varepsilon$, we have that $\sum_{n=2}^\infty Z^n = \frac{Z^2}{1-Z} < \frac{\varepsilon^2}{1-\varepsilon}$. Considering $g(Y) \le \sum_{n=2}^\infty Z^n$ and $\mathbb{E}\{ \sum_{n=2}^\infty Z^n \} < \frac{\varepsilon^2}{1-\varepsilon} < \infty$, which satisfies the DCT condition, yields that $\lim_{Q \to \infty} \mathbb{E}\{ g(Y) \} = \mathbb{E}\{\lim_{Q \to \infty} g(Y)\}$.}, and where $(b)$ follows from using the DCT to exchange the limitation and infinite summation operators in $\lim\limits_{Q \to \infty} g(Y)$ (similar proof method to the step $(a)$) and then by considering that $Y-\mathbb{E}\{Y\} \to 0$ as $Q \to \infty$ (due to the law of large numbers).
	By using similar mathematical manipulations, we also have that \begin{align}\label{limit2_Thm1_proof}
	\lim\limits_{Q = cL \to \infty} \mathbb{E}\left\{ X+Y \right\} \mathbb{E}\left\{ (X+Y)^{-1} \right\} = 1.
	\end{align}
	Considering the two limits \eqref{limit1_Thm1_proof} and \eqref{limit2_Thm1_proof}, we can directly conclude that $\lim\limits_{Q=cL \to \infty} \Delta  = 0$, and therefore prove \eqref{identity_proof}.
	
	Finally, substituting $\mathbb{E}\{X\} = \frac{P_t}{c G}(1+\frac{1}{L})$  and $\mathbb{E}\{Y\} = 1 +  \frac{P_t}{G} \frac{Q-1}{Q}$ into \eqref{identity_proof} and considering $Q=cL\to \infty$, completes the proof of Theorem~\ref{Tight_Bound_MF_coro}.

\section{Proof of Theorem~\ref{Rate_RZF_lemma}}\label{proof_theo_asymp}

We split the proof in three parts. First, we present the proof of~\eqref{SINR_L1_def}. Afterward, we provide two useful lemmas, and then we conclude by deriving the asymptotic deterministic equivalent of the~SINR. 

\vspace{-0.5cm}
\subsection{Proof of~\eqref{SINR_L1_def}}\label{app_sinr}
We provide here the proof of the expression of ${\rm SINR}_{\psi,k}^{\rm RZF}$ in~\eqref{SINR_L1_def}. 
Let us recall that ${\bf H}_{\psi,-k}$ represents the matrix ${\bf H}_{\psi}$ after removing its $k$-th row.  
The useful signal contribution to the received signal in~\eqref{y_RZF_Exp} (omitting the term $\nicefrac{\rho_\psi}{\sqrt{G}}$ for the sake of conciseness) can be written as
\begin{align}\label{signal_useful}
	&{\bf h}_{\psi,k}^T \left( \alpha {\bf I}_L + {\bf H}_\psi^H {\bf H}_\psi \right)^{-1} {\bf h}_{\psi,k}^* s_{\psi,k}
	\ = \ 
	{\bf h}_{\psi,k}^T  \Big( \alpha {\bf I}_L + {\bf H}_{\psi,-k}^H {\bf H}_{\psi,-k} + {\bf h}_{\psi,k}^* {\bf h}_{\psi,k}^T \Big)^{-1} {\bf h}_{\psi,k}^* s_{\psi,k} \notag\\
	&\hspace{3.2cm} \overset{(a)}{=}\  
		\frac{{\bf h}_{\psi,k}^T  \Big(\alpha {\bf I}_L + {\bf H}_{\psi,-k}^H {\bf H}_{\psi,-k}\Big)^{-1} {\bf h}_{\psi,k}^*}{ 1 +  {\bf h}_{\psi,k}^T \left(\alpha {\bf I}_L + {\bf H}_{\psi,-k}^H {\bf H}_{\psi,-k}\right)^{-1} {\bf h}_{\psi,k}^*} s_{\psi,k} 
	\ \overset{(b)}{=} \  
	\frac{A_{\psi,k}}{1+A_{\psi,k}} s_{\psi,k},
\end{align}
where $(a)$ follows from the relation 
\begin{align}\label{Matrix_identity}
		\big( {\bf A} - {\bf B} {\bf D}^{-1} {\bf C} \big)^{-1} {\bf B} {\bf D}^{-1}={\bf A}^{-1} {\bf B} \big( {\bf D} - {\bf C} {\bf A}^{-1} {\bf B}\big)^{-1},
\end{align}
and where $(b)$ follows after applying the definition of $A_{\psi,k}$ from~\eqref{def_a}. 

On the other hand, the power of the interference averaged over data signals in~\eqref{y_RZF_Exp} is given by
\begin{align}
		\big|I_{\psi,k}\big|^2
		&=  \frac{\rho_\psi^2}{G}   \sum_{\substack{\vartheta=1\\   \vartheta \neq k}}^L
		\sum_{\substack{\vartheta^\prime=1\\   \vartheta^\prime \neq k}}^L  {\bf h}_{\psi,\vartheta}^T  \left( \alpha {\bf I}_L + {\bf H}_\psi^H {\bf H}_\psi \right)^{-1} {\bf h}_{\psi,k}^* 
		{\bf h}_{\psi,k}^T \left( \alpha {\bf I}_L + {\bf H}_\psi^H {\bf H}_\psi \right)^{-1} {\bf h}_{\psi,\vartheta^\prime}^*  \mathbb{E} \{s_{\psi,\vartheta}^* s_{\psi,\vartheta^\prime} \} \notag\\
		&= \frac{\rho_\psi^2}{G}
		 {\bf h}_{\psi,k}^T \left( \alpha {\bf I}_L + {\bf H}_\psi^H {\bf H}_\psi \right)^{-1} {\bf H}_{\psi,-k}^H {\bf H}_{\psi,-k}
		 \left( \alpha {\bf I}_L + {\bf H}_\psi^H {\bf H}_\psi \right)^{-1} {\bf h}_{\psi,k}^*.
\end{align}
By applying again the matrix identity in~\eqref{Matrix_identity} and by considering the definitions of $A_{\psi,k}$ and $B_{\psi,k}$ in~\eqref{def_a}-\eqref{def_b},  we can obtain
	\begin{align} 
		\big|I_{\psi,k}\big|^2 \!
	    = \frac{\rho_\psi^2}{G} \frac{{\bf h}_{\psi,k}^T \! \left(\alpha {\bf I}_L + {\bf H}_{\psi,-k}^H {\bf H}_{\psi,-k} \right)^{-1} \! {\bf H}_{\psi,-k}^H {\bf H}_{\psi,-k} \big(\alpha {\bf I}_L + {\bf H}_{\psi,-k}^H {\bf H}_{\psi,-k}\big)^{-1} {\bf h}_{\psi,k}^*}{\left(1+{\bf h}_{\psi,k}^T \left(\alpha {\bf I}_L + {\bf H}_{\psi,-k}^H {\bf H}_{\psi,-k} \right)^{-1} {\bf h}_{\psi,k}^*\right)^2}
		\!=\!  \frac{B_{\psi,k} \rho_\psi^2/G}{(1+A_{\psi,k})^2} \notag
	\end{align}
which combined with~\eqref{signal_useful} yields the expression of ${\rm SINR}_{\psi,k}^{\rm RZF}$ in~\eqref{SINR_L1_def}. This concludes the proof.

\vspace{-0.5cm}
\subsection{Two Useful Lemmas} 
  
	In the following, we present two lemmas that are instrumental in the derivation of Lemma~\ref{Rate_RZF_lemma}. 
    \begin{lemma}\label{first_lemma_proof}
		For any fixed $c$, $0<c <\infty$, the trace of $\frac{1}{L} \left( z {\bf I}_L + \frac{1}{L}{\bf H}_\psi^H {\bf H}_\psi \right)^{-1}$ converges to $S_c(z)$ almost surely as $L\to\infty$, where $S_c(z)$ is defined as 
		\begin{align}
				S_c(z) &\triangleq \frac{1}{2} \bigg( \sqrt{ \frac{(1-c)^2}{z^2} + \frac{2 (1+c)}{z} +1} +  \frac{1-c}{z} -1 \bigg). \label{S_z_def}
		\end{align}		
    \end{lemma}
    \begin{proof}
		This lemma can be obtained as a direct application of a known result from~\cite[Ch. 3]{Debbah} for the Stieltjes transform\cite{widder1938stieltjes}. Hence, we omit the proof due to the page limitation and refer the reader to~\cite[Ch. 3]{Debbah} for more details.
    \end{proof}

    \begin{lemma}\label{second_lemma_proof}\vspace{-0.3cm}
		For any fixed $0<c<\infty$ and arbitrary $0<\theta<\infty $, we have that, as $L\to\infty$, 
		\begin{align}
		{\rm Tr} \left\{ \frac{1}{L} \left(\theta {\bf I} + \frac{1}{L} {\bf H}_{\psi,-k}^H  {\bf H}_{\psi,-k}\right)^{-2} \right\} 
		\stackrel{a.s.}{\longrightarrow} {\rm Tr} \left\{ \frac{1}{L} \left(\theta {\bf I} + \frac{1}{L} {\bf H}_{\psi}^H  {\bf H}_{\psi}\right)^{-2} \right\}.
		\end{align}
    \end{lemma}

    \begin{proof}
		Let us first define ${\bf A} \triangleq \theta {\bf I} + \frac{1}{L} {\bf H}_{\psi,-k}^H {\bf H}_{\psi,-k}$, and let us also define
		\begin{align}
		    \delta \triangleq \ &\ \Big| {\rm Tr} \Big\{ \frac{1}{L} \big(\theta {\bf I} + \frac{1}{L} {\bf H}_{\psi,-k}^H  {\bf H}_{\psi,-k}\big)^{-2} \Big\} 
		        - {\rm Tr} \Big\{ \frac{1}{L} \big(\theta {\bf I} + \frac{1}{L} {\bf H}_{\psi}^H  {\bf H}_{\psi}\big)^{-2} \Big\} \Big|.
		\end{align}
		By applying the Woodbury matrix identity\cite{woodbury1950inverting}, 
		we can rewrite $\delta$ as 
		\begin{align}\label{Delta_long}
		\delta 
		&=  \Big| \frac{1}{L} {\rm Tr} \Big\{  \frac{2}{L}\frac{ {\bf h}_k^T {\bf A}^{-3} {\bf h}_k^*}{1+ \frac{1}{L} {\bf h}_k^T {\bf A}^{-1} {\bf h}_k^*} - \frac{1}{L^2} \frac{ ({\bf h}_k^T {\bf A}^{-2} {\bf h}_k^*) ({\bf h}_k^T {\bf A}^{-2} {\bf h}_k^*)}{\left(1+ \frac{1}{L} {\bf h}_k^T {\bf A}^{-1} {\bf h}_k^*\right)^2} \Big\} \Big|,
		\end{align}
		which can be further rewritten as
		$
		\delta = \big| \Theta_1 - \Theta_2 \big|, 
		$
		where
			$\Theta_1 \triangleq   \frac{2}{L^2} \frac{ {\bf h}_k^T {\bf A}^{-3} {\bf h}_k^*}{1+ \frac{1}{L} {\bf h}_k^T {\bf A}^{-1} {\bf h}_k^*}$,  and 
			$ \Theta_2 \triangleq \frac{1}{L^3} \left( \frac{{\bf h}_k^T {\bf A}^{-2} {\bf h}_k^*} {1+\frac{1}{L} {\bf h}_k^T {\bf A}^{-1} {\bf h}_k^*}\right)^2$. 
		Furthermore, we can apply eigenvalue decomposition by factorizing $\frac{1}{L} {\bf H}_{\psi,-k}^H {\bf H}_{\psi,-k}$ as $\frac{1}{L} {\bf H}_{\psi,-k}^H {\bf H}_{\psi,-k} = {\bf Q} {\bf \Lambda} {\bf Q}^H$, which yields ${\bf A}^{-1} =  {\bf Q} \left( \theta {\bf I} + {\bf \Lambda} \right)^{-1} {\bf Q}^H$, and ${\bf A}^{-3} =  {\bf Q} \left(\theta {\bf I} + {\bf \Lambda} \right)^{-3} {\bf Q}^H$.
		Thus, upon defining ${\bf g} \triangleq {\bf Q}{\bf h}^*_k/\sqrt{L}$, the term  $\Theta_1$ can be rewritten as
			\begin{align}\label{Theta_1_limit}
				\Theta_1 &
				= \frac{2}{L^2} \frac{ {\bf h}_k^T {\bf Q} \left(\theta {\bf I} + {\bf \Lambda} \right)^{-3} {\bf Q}^H {\bf h}_k^*}{1+ \frac{1}{L} {\bf h}_k^T {\bf Q} \left(\theta {\bf I} + {\bf \Lambda} \right)^{-1} {\bf Q}^H {\bf h}_k^*} 
				=\frac{2}{L} \frac{ {\bf g}^H \left(\theta {\bf I} + {\bf \Lambda} \right)^{-3} {\bf g}}{1+ {\bf g}^H \left(\theta {\bf I} + {\bf \Lambda} \right)^{-1} {\bf g}} \notag\\
				& = \frac{2}{L}\frac{\sum_{\ell=1}^L |g_\ell|^2 \frac{1}{(\theta+\lambda_\ell)^3}}{1+\sum_{\ell=1}^L |g_\ell|^2 \frac{1}{\theta+\lambda_\ell}}
				\le \frac{2}{\theta^2 L}\frac{\sum_{\ell=1}^L |g_\ell|^2 \frac{1}{\theta+\lambda_\ell}}{1+\sum_{\ell=1}^L |g_\ell|^2 \frac{1}{\theta+\lambda_\ell}} 
				\le \frac{2}{\theta^2 L} \to 0, \text{ as } L \to \infty,
			\end{align}
		where $g_\ell$ and $\lambda_\ell$ are the $\ell$-th element of $\bf g$ and the $\ell$-th eigenvalue of $\frac{1}{L}{\bf H}_{\psi,k}^H {\bf H}_{\psi,k}$, respectively. 
		Similarly, 
		we have  that
			\begin{align}\label{Theta_2_limit}
				&\Theta_2 
				\!=\! \frac{1}{L} \left(\frac{{\bf g}^H (\theta {\bf I}+ {\bf \Lambda})^{-2} {\bf g}}{1+{\bf g}^H (\theta {\bf I}+{\bf \Lambda})^{-1} {\bf g}}\right)^2 
				\!\le\! \frac{1}{\theta^2 L} \left( \frac{\sum_{\ell=1}^L |g_\ell|^2 \frac{1}{\theta+\lambda_\ell}}{1+\sum_{\ell=1}^L |g_\ell|^2 \frac{1}{\theta+\lambda_\ell}} \right)^2 
				\!\le\! \frac{1}{\theta^2 L} \to 0, \text{ as } L \to \infty.
			\end{align}
		Finally, from \eqref{Theta_1_limit}, \eqref{Theta_2_limit}, and from the fact that $\delta \le \big| \Theta_1\big|+ \big|\Theta_2 \big|$, 
		the difference $\delta$ approaches zero almost surely as $L \to \infty$. This concludes the proof of Lemma~\ref{second_lemma_proof}.\end{proof}

    \vspace{-0.5cm}\subsection{Proof of  Theorem~\ref{Rate_RZF_lemma}}\label{proof_lem_asymp}

        We obtain Theorem~\ref{Rate_RZF_lemma} 
        by deriving the asymptotic deterministic equivalent  of ${\rm SINR}_{\psi,k}$ in~\eqref{SINR_L1_def}. For that, we first derive the asymptotic deterministic equivalent of~$A_{\psi,k}$ and~$\rho_\psi^2$. 
	
		Let us start by considering $A_{\psi,k}$, defined in~\eqref{def_a}. By means of the Trace Lemma and the Rank-1 Perturbation Lemma from~\cite{Debbah}, we can obtain that
		\begin{align}\label{A_appendix}
			A_{\psi,k} &= {\bf h}_{\psi,k}^T  \Big(\alpha {\bf I}_L + {\bf H}_{\psi,-k}^H {\bf H}_{\psi,-k}\Big)^{-1} {\bf h}_{\psi,k}^*  
			 \stackrel{a.s.}{\longrightarrow}  {\rm Tr} \left\{ \Big(\alpha {\bf I}_L + {\bf H}_{\psi}^H {\bf H}_{\psi}\Big)^{-1} \right\}
		\end{align}
		as $L \to \infty$. 
		From this, we can apply Lemma~\ref{first_lemma_proof} and the fact that $\alpha = L/P_t$ to obtain the deterministic equivalent of $A_{\psi,k}$, which we denote as $a_{\psi,k}$, and which is given by 
		 $a_{\psi,k} = S_c\left( \frac{1}{P_t} \right)$, where 	$	S_c(z) = \frac{1}{2} \Big[ \sqrt{ \frac{(1-c)^2}{z^2} + \frac{2 (1+c)}{z} +1} +  \frac{1-c}{z} -1 \Big]$ as defined in~\eqref{S_z_def}. This yields the expression of~$a_{\psi,k}$ in~\eqref{A_psi_K_deter}. 
		
		Next, we focus on $B_{\psi,k}$, introduced in~\eqref{def_b}, and we again apply the Trace Lemma and the Rank-1 Perturbation Lemma from~\cite{Debbah} in the limit of $L\to\infty$ to obtain that
		\begin{align}\label{B_appendix}
			B_{\psi,k} &  \stackrel{a.s.}{\longrightarrow} \frac{1}{L} {\rm Tr} \Big\{ \frac{1}{L} {\bf H}_{\psi,-k}^H {\bf H}_{\psi,-k} \Big(\frac{1}{P_t} {\bf I}_L + \frac{1}{L} {\bf H}_{\psi,-k}^H {\bf H}_{\psi,-k}\Big)^{-2} \Big\} \notag\\
			&= \frac{1}{L}{\rm Tr} \Big\{ \Big(\frac{1}{P_t} {\bf I}_L + \frac{1}{L}{\bf H}_{\psi,-k}^H {\bf H}_{\psi,-k}\Big)^{-1} \Big\}
			- \frac{1}{P_tL}{\rm Tr} \Big\{ \Big(\frac{1}{P_t} {\bf I}_L + \frac{1}{L} {\bf H}_{\psi,-k}^H {\bf H}_{\psi,-k}\Big)^{-2} \Big\}.
		\end{align}
	The first trace term of the R.H.S. of~\eqref{B_appendix}  matches~\eqref{A_appendix}, and thus its deterministic equivalent is $a_{\psi,k}$. 
	With respect to the second term of the R.H.S. of~\eqref{B_appendix}, applying Lemmas~\ref{first_lemma_proof} and~\ref{second_lemma_proof} yields
		\begin{align}\label{second_part_B}
			\frac{1}{P_tL} & {\rm Tr} \Big\{ \Big(\frac{1}{P_t} {\bf I}_L + \frac{1}{L} {\bf H}_{\psi,-k}^H {\bf H}_{\psi,-k}\Big)^{-2} \Big\} \stackrel{a.s.}{\longrightarrow} \frac{1}{P_t L} {\rm Tr}\Big\{ \Big(\frac{1}{P_t}{\bf I}_L + \frac{1}{L} {\bf H}_{\psi}^H  {\bf H}_{\psi}\Big)^{-2} \Big\} \notag\\
			&= \frac{1}{P_tL} \sum\nolimits_{\ell=1}^L \frac{1}{(\lambda_\ell+1/P_t)^2}
			= -\frac{1}{P_t} \frac{\partial}{\partial z} \Big( \frac{1}{L} \sum\nolimits_{\ell=1}^{L} \frac{1}{\lambda_\ell+z} \Big)\Big|_{z=1/P_t}\notag\\
			& = -\frac{1}{P_t} \frac{\partial}{\partial z} \Big(  {\rm Tr} \Big\{ \frac{1}{L} \Big(z {\bf I}_L + \frac{1}{L} {\bf H}_\psi^H {\bf H}_\psi \Big)^{-1} \Big\} \Big)\Big|_{z=1/P_t}
			\stackrel{a.s.}{\longrightarrow} -\frac{1}{P_t} \frac{\partial S_c(z)}{\partial z} \Big|_{z=1/P_t},
		\end{align}
		as $L\to\infty$,  where $\{\lambda_\ell\}_{\ell=1}^L$ are the eigenvalues of $\frac{1}{L} {\bf H}_\psi^H {\bf H}_\psi$ and where $\frac{\partial S_c(z)}{\partial z}$ is the derivative of $S_c(z)$ with respect to $z$, which is given by	
		\begin{align}
				\frac{\partial S_c(z)}{\partial z}&= \frac{1}{2} \left[ \frac{-c^2-c(z-2)-z-1}{z^2 \sqrt{c^2+2c(z-1)+(z+1)^2}} -\frac{1-c}{z^2} \right].\label{eq:def_derivate_sz}
		\end{align}	
		From \eqref{B_appendix} and  \eqref{second_part_B} it holds that 
		\begin{align}\label{eq:ade_B}
				B_{\psi,k} \stackrel{a.s.}{\longrightarrow} b_{\psi,k}
				 \triangleq a_{\psi,k}+\frac{1}{P_t} \frac{\partial S_c(z)}{\partial z} \Big|_{z=1/P_t}\qquad \text{as }L\to\infty.
		\end{align}
	To conclude, we focus on the power control factor for the RZF precoder, which was given by 
	$\rho_\psi^2 = \frac{P_t}{\frac{1}{L}{\rm Tr} \{ \frac{1}{L} {\bf H}_\psi^H  {\bf H}_\psi ( \frac{1}{L}{\bf H}_\psi^H {\bf H}_\psi + \frac{1}{P_t} {\bf I}_L)^{-2} \}}$.
	In view of the derivation of $B_{\psi,k}$ and \eqref{B_appendix}--\eqref{second_part_B}, it follows that $\rho_\psi^2 \stackrel{a.s.}{\longrightarrow}  \frac{P_t}{b_{\psi,k}}$. 
    Thus, we have that the asymptotic deterministic equivalent of $\rho_\psi^2$, denoted by $p_\psi^2$, takes the form $p_\psi^2 = \frac{P_t}{b_{\psi,k}}$, 
    which, upon substituting~\eqref{eq:def_derivate_sz} in $b_{\psi,k}$, yields the expression of $p_\psi^2$ in~\eqref{rho_psi_K_deter}.
	
	Next, we obtain the  asymptotic deterministic equivalent of ${\rm SINR}_{\psi,k}$ by substituting the asymptotic deterministic equivalent of $A_{\psi,k}$, $B_{\psi,k}$ and $\rho_\psi^2$ into~\eqref{SINR_L1_def}, which yields
		\begin{align}\label{SINR_L1_def:asymp}
			{\rm SINR}_{\psi,k}^{\rm RZF}  \stackrel{a.s.}{\longrightarrow}
			\frac{a_{\psi,k}^2{p_\psi^2}/{G} }{\big(1+a_{\psi,k} \big)^2+\frac{P_t}{G}}.
		\end{align}	
	Finally, a direct application of the Continuous Mapping Theorem \cite{Vaart_book} yields~\eqref{Rate_RZF_eq}, which concludes the proof of Theorem~\ref{Rate_RZF_lemma}. \qed

	\bibliographystyle{IEEEtran}				
	\bibliography{IEEEabrv,aBiblio}			
\end{document}